\documentclass[a4paper, 11pt]{article}

\usepackage{microtype}
\usepackage{graphicx}
\usepackage{subfigure}
\usepackage{booktabs}

\usepackage[hidelinks]{hyperref}

\usepackage{amsmath}
\usepackage{amssymb}
\usepackage{mathtools}
\usepackage{amsthm}

\usepackage[capitalize,noabbrev,nameinlink]{cleveref}

\usepackage{algorithm}
\usepackage{algorithmic}
\usepackage{fullpage}

\usepackage{centernot}

\newtheorem{innercustomthm}{Theorem}

\theoremstyle{plain}
\newtheorem{theorem}{Theorem}[section]

\newtheorem{lemma}[theorem]{Lemma}
\newtheorem{corollary}[theorem]{Corollary}
\newtheorem{observation}[theorem]{Observation}
\newtheorem{claim}[theorem]{Claim}
\theoremstyle{definition}
\newtheorem{definition}[theorem]{Definition}

\theoremstyle{remark}
\newtheorem{remark}[theorem]{Remark}

\newenvironment{customthm}[1]
{\renewcommand\theinnercustomthm{#1}\innercustomthm}
  {\endinnercustomthm}

\DeclareMathOperator*\Exp{\bf E}
\DeclareMathOperator*\Prob{\bf Pr}

\renewcommand{\P}{\mathsf{P}}

\newcommand{\ignore}[1]{}
\newcommand{\prn}[1]{\left(#1\right)}
\newcommand{\cprn}[1]{\!\left(#1\right)}
\newcommand{\sqbra}[1]{\left[#1\right]}
\newcommand{\csqbra}[1]{\!\left[#1\right]}

\newcommand{\abs}[1]{\left|#1\right|}

\newcommand{\brkts}[1]{\left\{#1\right\}}

\newcommand{\bool}{\brkts{0,1}}

\newcommand{\cpoly}[1]{\textup{poly}\cprn{#1}}

\newcommand{\pr}[1]{\Prob\sqbra{#1}}
\newcommand{\cpr}[1]{\Prob\csqbra{#1}}

\newcommand{\cprq}[2]{\Prob_{#2}\csqbra{#1}}

\newcommand{\cexpect}[1]{\Exp\csqbra{#1}}

\newcommand{\cexpectq}[2]{\Exp_{#2}\csqbra{#1}}

\newcommand{\R}{\mathbb{R}}
\newcommand{\NP}{\mathsf{NP}}

\newcommand{\dtv}{d_{\mathrm{TV}}}
\newcommand{\U}{\mathbb{U}}
\newcommand{\SZK}{\mathsf{SZK}}
\newcommand{\NISZK}{\mathsf{NISZK}}

\allowdisplaybreaks

\title{Total Variation Distance Meets Probabilistic Inference%
\thanks{The author list has been sorted alphabetically by last name; this should not be used to determine the extent of authors’ contributions.
This work has been accepted for presentation at the International Conference on Machine Learning (ICML) 2024.}
}

\author{%
{\bf Arnab Bhattacharyya}\\
National University of Singapore
\and
{\bf Sutanu Gayen}\\
Indian Institute of Technology Kanpur
\and
\bf Kuldeep S. Meel\\
University of Toronto
\and
\bf Dimitrios Myrisiotis\\
CNRS@CREATE LTD.
\and
\bf A. Pavan\\
Iowa State University
\and
\bf N. V. Vinodchandran\\
University of Nebraska-Lincoln
}

\begin{document}

\maketitle

\begin{abstract}
In this paper, we establish a novel connection between total variation (TV) distance estimation and probabilistic inference.
In particular, we present an efficient, structure-preserving reduction from relative approximation of TV distance to probabilistic inference over directed graphical models.
This reduction leads to a fully polynomial randomized approximation scheme (FPRAS) for estimating TV distances between same-structure distributions over any class of Bayes nets for which there is an efficient probabilistic inference algorithm.
In particular, it leads to an FPRAS for estimating TV distances between distributions that are defined over a common Bayes net of small treewidth.
Prior to this work, such approximation schemes only existed for estimating TV distances between product distributions.
Our approach employs a new notion of \emph{partial} couplings of high-dimensional distributions, which might be of independent interest.
\end{abstract}

\section{Introduction}

\label{sec:intro}

Substantial research has been devoted to developing models that represent high-dimensional probability distributions succinctly.
One prevalent approach is through graphical models.
In a graphical model, a graph describes the conditional dependencies among variables and the probability distribution is factorized according to the adjacency relationships in the graph~\cite{koller2009probabilistic}.
When the underlying graph is a directed graph, the model is known as a Bayesian network or Bayes net.

Two fundamental computational tasks on distributions are {\em distance computation} and {\em probabilistic inference}.
In this work, we establish a novel connection between these two seemingly different computational tasks.
Using this connection, we design new relative error approximation algorithms for estimating the statistical distance between Bayes net distributions with small treewidth.

\paragraph{Total Variation Distance Computation.}

The distance computation problem is the following:
Given descriptions of two probability distributions $P$ and $Q$, compute $\rho(P,Q)$ for a distance measure $\rho$.
A distance measure of central importance is the {\em total variation (TV) distance} (also known as the {\em statistical distance}).
Let $P$ and $Q$ be distributions over a finite domain $\mathcal{D}$.
The total variation distance between $P$ and $Q$, denoted by $\dtv(P,Q)$, is defined as
\[
\dtv(P,Q)
= \max_{S \subseteq \mathcal{D}}\cprn{P(S)-Q(S)}.
\]
The total variation distance satisfies many basic properties which makes it a versatile and fundamental measure for quantifying the dissimilarity between probability distributions. 
First, it has an explicit probabilistic interpretation:
The TV distance between two distributions is the maximum gap between the probabilities assigned to a single event by the two distributions.
Second, it satisfies many mathematically desirable properties:
It is bounded and lies in $[0,1]$, it is a metric, and it is invariant with respect to bijections.
Total variation distance also measures the minimum probability that $X\neq Y$ among all couplings $(X,Y)$ between $P$ and $Q$.
Because of these reasons, the total variation distance is a central distance measure employed in a wide range of areas including probability and statistics~\cite{mitzenmacher2005probability}, machine learning~\cite{shwartz2014understanding}, information theory~\cite{cover2006elements}, cryptography~\cite{stinson1995cryptography}, data privacy~\cite{dwork2006differential}, and pseudorandomness~\cite{vadhan2012pseudorandomness}.

\paragraph{Probabilistic Inference.}

Probabilistic inference in graphical models is a fundamental computational task with a wide range of applications that spans disciplines including statistics, machine learning, and artificial intelligence (e.g., see~\cite{wainwright2008graphical}).
Various algorithms have been proposed for this problem, encompassing both exact approaches like message passing~\cite{pearl1988probabilistic}, variable elimination~\cite{dechter1999bucket}, and junction-tree propagation~\cite{lauritzen1988local}, as well as approximate techniques such as loopy belief propagation, variational inference-based methods~\cite{wainwright2008graphical}, and particle-based algorithms (refer to Chapter 13 of~\cite{koller2009probabilistic} and the references therein).
Computational hardness results have also been established in several works~\cite{cooper1990computational, littman2001stochastic, roth1996hardness, kwisthout2010necessity}.
In this work, we rely on the following formulation of probabilistic inference:
Given (a representation of) random variables $X_1,\dots,X_n$ and (a representation of) sets $S_1,\dots,S_n$ such that for all $i$ the set $S_i$ is a subset of the range of $X_i$, compute
\[
\cpr{X_1\in S_1,\dots,X_n\in S_n}.
\]

\subsection{Our Contributions}

The problems of total variation distance computation and probabilistic inference have been studied for nearly four decades on their own, but there is no known relationship between these two fundamental yet seemingly different computational tasks.
The primary goal of our paper is to initiate an investigation to determine such a relationship.
Surprisingly, we demonstrate that there is a {\em structure-preserving} reduction from the TV distance estimation problem to the probabilistic inference problem over Bayes nets:
In particular, we exhibit an efficient probabilistic reduction that, given two Bayes nets $P$ and $Q$ defined over a directed acyclic graph (DAG) $G$, makes probabilistic inference queries to a Bayes net $\mathcal{L}$ defined over the {\em same} DAG $G$ and returns a relative approximation of $\dtv(P, Q)$.

\begin{theorem}[Informal]
\label{thm:dtv-to-inference}
There is a polynomial-time randomized algorithm that takes a DAG $G$, two Bayes nets $P$ and $Q$ over $G$, and parameters $\varepsilon,\delta$ as inputs and behaves as follows.
The algorithm makes probabilistic inference oracle queries to a Bayes net over the same DAG $G$ and outputs a $(1+\varepsilon)$-relative approximation of $\dtv(P,Q)$ with probability at least $1-\delta$.
\end{theorem}

It is known that probabilistic inference computation over Bayes nets in general is a $\#\P$-hard problem and hence exact $\dtv$ computation reduces to probabilistic inference over Bayes nets~\cite{cooper1990computational}.
However, a salient feature of our reduction is that it {\em preserves the structure} of the Bayes net.
Note that exact $\dtv$ computation is $\#\P$-complete even for product distributions for which inference computation is straightforward~\cite{BGMMPV23}.

A conceptual contribution of our work is the introduction of a new notion of {\em partial coupling} between two probability distributions, which is a relaxation of the classical notion of coupling (Definition~\ref{def:partial-coupling}).
Specifically, we illustrate that while establishing a computationally efficient coupling for distributions such as Bayesian networks may not be possible, it is possible to define a computationally efficient {\em partial coupling}.
Remarkably, we show that a {\em partial coupling} is adequate for approximating the total variation distance.
The technique of coupling, introduced by Doeblin in 1938~\cite{Doe38}, has been fundamental in the realms of computer science and statistics for over four decades, underpinning some of the most seminal results~\cite{Lin02,LevinPeresWilmer2006,MT12}.
In a similar vein, we believe the notion of {\em partial coupling} possesses the potential to become an essential tool in the toolkit of these domains.

The aforementioned reduction from total variation distance and probabilistic inference leads to efficient $\dtv$ estimation algorithms for any class of Bayes nets that admits efficient probabilistic inference algorithms.
In particular, it leads to the first polynomial-time randomized approximation scheme for calculating the total variation distance between two Bayes nets of treewidth $O(\log n)$, since the well-known variable elimination algorithm can be used for efficient probabilistic inference for such Bayes nets.

\begin{theorem}[Informal]
\label{thm:FPRAS-bounded-tw}
There is an FPRAS for estimating the TV distance between two Bayes nets of treewidth $O(\log n)$ that are defined over the same DAG of $n$ nodes.
\end{theorem}

Prior to our work, such approximation schemes were known only for product distributions, which are Bayes nets over a graph with no edges~\cite{fgjw22}.
In particular, designing an FPRAS for estimating TV distance between Bayes nets over trees (which are graphs with treewidth $1$) was an open question.
Our result resolves this open question.
In fact, \Cref{thm:FPRAS-bounded-tw} shows that it is indeed possible to obtain an FPRAS for a large class of Bayes nets, namely Bayes nets of $O(\log n)$ treewidth.

Note that the setting of \Cref{thm:FPRAS-bounded-tw} (whereby the Bayes net distributions considered are over the same DAG) is practically relevant where one learns parameters of a Bayes net for a fixed structure from different batches and a natural question is whether the two models are close to each other or not.
The TV distance-based approaches have played a significant role in the testing and improvement of constrained samplers~\cite{golia2021designing}.

Our next set of results focuses on the case when one of the distributions is the uniform distribution.
We first prove that the exact computation of the TV distance between a Bayes net distribution and the uniform distribution is $\#\P$-complete.

\begin{theorem}
\label{thm:sharp-P-complete}
It is $\#\P$-complete to exactly compute the TV distance between a Bayes net that has bounded in-degree and the uniform distribution.
\end{theorem}

To complement this result, we show that there is an FPRAS that estimates the TV distance between the uniform distribution and {\em any} Bayes net distribution.

\begin{theorem}[Informal]
\label{thm:FPRAS-bayes-net-vs-uniform}
There is an FPRAS for estimating the TV distance between a Bayes net and the uniform distribution.
\end{theorem}

\subsection{Related Work}

Koller and Friedman~\cite{koller2009probabilistic} provide a comprehensive overview of probabilistic graphical models (such as Bayes nets).

\paragraph{TV Distance Computation.}

Recently,~\cite{BGMMPV23} initiated the study of the computational complexity aspects of TV distance over graphical models.
In that work, they proved that exactly computing the TV distance between product distributions is $\#\P$-complete, that it is $\NP$-hard to decide whether the TV distance between two Bayes nets of in-degree $2$ is equal to $0$ or not, and also gave an FPTAS for approximating the TV distance between an arbitrary product distribution and any product distribution that has a constant number of distinct marginals (note that this includes the uniform distribution).
In a subsequent work,~\cite{fgjw22} gave an FPRAS for approximating the TV distance between two arbitrary product distributions.
Later,~\cite{feng2023ondeterministically} gave a \emph{deterministic} approximation algorithm (FPTAS) for the same task.

TV distance estimation was also studied previously from a more complexity-theoretic and cryptographic viewpoint.
\cite{SV03} established in a seminal work that additively approximating the TV distance between two distributions that are samplable by Boolean circuits is hard for $\SZK$ (Statistical Zero Knowledge).
\cite{GoldreichSV99} showed that the problem of deciding whether a distribution samplable by a Boolean circuit is close or far from the uniform distribution is complete for the complexity class $\NISZK$ (Non-Interactive Statistical Zero Knowledge).

Additive approximation of TV distance is much easier.
\cite{CR14} showed how to additively estimate TV distance between distributions that can be efficiently sampled and whose probability mass functions can be efficiently evaluated.
Clearly, Bayes nets satisfy both conditions (where ``efficient'' means as usual polynomial in the number of parameters).
\cite{0001GMV20} extended this idea to develop polynomial-time algorithms for additively approximating the TV distance between two bounded in-degree Bayes nets using a polynomial number of samples from each.

\paragraph{Probabilistic Inference.}

There is a significant body of work dedicated to exact probabilistic inference.
As we mentioned earlier, some algorithmic paradigms that have been developed for the task of probabilistic inference are message passing~\cite{pearl1988probabilistic}, variable elimination~\cite{dechter1999bucket}, and junction-tree propagation~\cite{lauritzen1988local}.
Recently,~\cite{klinkenberg2023exactbayesian} presented an exact Bayesian inference method for inferring posterior distributions encoded by probabilistic programs.
\cite{zaiser2023exact} present an exact Bayesian inference method for discrete statistical models, by introducing a probabilistic programming language (based on probability generating functions) that supports discrete and continuous sampling, and conditioning on discrete events (among others).
\cite{holtzen2020dice} develop a domain-specific probabilistic programming language, called Dice, that exploits program structure in order to factorize inference, enabling them to perform exact inference on large probabilistic programs.
\cite{saad2021sppl} present the Sum-Product Probabilistic Language (SPPL), a new probabilistic programming language that automatically delivers exact solutions to a broad range of probabilistic inference queries enabling them to give exact algorithms for conditioning on and computing probabilities of events.

With the advent of big data and the increasing complexity of models, traditional exact inference methods may become computationally infeasible.
Approximate inference techniques, such as variational inference and sampling methods like Markov Chain Monte Carlo, provide efficient and scalable alternatives to tackle these challenges.
\cite{minka2001expectation} introduces the expectation propagation algorithm for approximate Bayesian inference.
\cite{hoffman2013stochastic} propose a stochastic variational inference algorithm for large-scale Bayesian inference.
\cite{murphy2013loopy} investigate the effectiveness of loopy belief propagation.
\cite{ranganath2014black} introduce black box variational inference, a flexible and scalable approach for approximate Bayesian inference.
\cite{rezende2015variational} propose a variational inference method using normalizing flows, a class of flexible and expressive transformations.
\cite{blei2017variational} provide a comprehensive review of variational inference, a family of methods for approximate Bayesian inference.

\subsection{Organization}

The rest of the paper is organized as follows.
We provide some background material in \Cref{sec:preliminaries} and a technical overview of our results in \Cref{sec:technical}.
We prove the main results as follows:
We show \Cref{thm:dtv-to-inference} in \Cref{sec:dtv-to-inference}; \Cref{thm:FPRAS-bounded-tw} in \Cref{sec:FPRAS-bounded-tw};
\Cref{thm:sharp-P-complete} in \Cref{sec:hardness};
\Cref{thm:FPRAS-bayes-net-vs-uniform} in \Cref{sec:Bayes-nets-vs-Uniform}.
We conclude in \Cref{sec:conclusion}.

\section{Preliminaries}

\label{sec:preliminaries}

We use $\sqbra{n}$ to denote the set $\brkts{1,\dots,n}$ and $\log$ to denote $\log_2$.
Throughout the paper, we shall assume that all probabilities are represented as rational numbers of the form $a/b$.
We denote the uniform distribution by $\U$.

The following concentration inequality will be useful in our proofs.

\begin{lemma}[Hoeffding's inequality]
\label{lem:Hoeffding}
Let $X_1,\dots,X_n$ be independent random variables such that $a_{i}\leq X_{i}\leq b_{i}$ for all $1\leq i\leq n$.
Then
\[
\cpr{\abs{\sum_{i=1}^nX_i-\cexpect{\sum_{i=1}^nX_i}}\geq t}
\leq 2\exp\!\left(-2t^2/\sum_{i=1}^{n}(b_{i}-a_{i})^{2}\right).
\]
\end{lemma}

We shall use the following notion of an approximation algorithm.

\begin{definition}[FPRAS]
A function $f: \bool^* \to \R$ admits a \emph{fully polynomial-time randomized approximation scheme (FPRAS)} if there is a \emph{randomized} algorithm $\mathcal{A}$ such that for all $n$ and all inputs $x\in\bool^n$, $\varepsilon>0$, and $\delta>0$, $\mathcal{A}$ outputs a $\prn{1+\varepsilon}$-relative approximation of $f(x)$, i.e., a value $v$ that lies in the interval $[f(x)/(1+\varepsilon), (1+\varepsilon)f(x)]$, with probability $1-\delta$.
The running time of $\mathcal{A}$ is polynomial in $n,1/\varepsilon,1/\delta$.
\end{definition}

\subsection{Bayes Nets}

For a directed acyclic graph (DAG) $G$ and a node $v$ in $G$, let $\Pi(v)$ denote the set of parents of $v$.

\begin{definition}[Bayes nets]
A \emph{Bayes net} is specified by a directed acyclic graph (DAG) over a vertex set $[n]$ and a collection of probability distributions over symbols in $\sqbra{\ell}$, as follows.
Each vertex $i$ is associated with a random variable $X_i$ whose range is $[\ell]$.
Each node $i$ of $G$ has a Conditional Probability Table (CPT) that describes the following:
For every $x \in [\ell]$ and every $y \in [\ell]^k$, where $k$ is the size of $\Pi(i)$, the CPT has the value of $\Prob[X_i =x|X_{\Pi\prn{i}} = y]$ stored.
Given such a Bayes net, its associated probability distribution $P$ is given by the following:
For all $x\in [\ell]^n$, $P\cprn{x}$ is equal to
\[
\cprq{X = x}{P}
=\prod_{i=1}^n\cprq{X_i=x_i|X_{\Pi\prn{i}}=x_{\Pi\prn{i}}}{P}.
\]
Here $X$ is the joint distribution $(X_1,\ldots,X_n)$ and $x_{\Pi\prn{i}}$ is the projection of $x$ to the indices in $\Pi(i)$.
\end{definition}

Note that $P(x)$ can be computed in linear time by using the CPTs of $P$ to retrieve each $\cprq{X_i=x_i|X_{\Pi\prn{i}}=x_{\Pi\prn{i}}}{P}$.
We also use $P_{i|\Pi(i)}\cprn{x_i|x_{\Pi(i)}}$ to denote this probability.

An important notion is that of the moralization of a Bayes net.

\begin{definition}[Moralization of Bayes nets]
\label{def:moralization}
Let $B$ be a Bayes net over a DAG $G$.
The \emph{moralization of $B$} is the undirected graph that is obtained from $G$ as follows.
For every node $u$ of $G$ and any pair $\prn{v,w}$ of its parents $\Pi\cprn{u}$ if $v$ and $w$ are not connected by some edge in $G$, then add the edge $\prn{v,w}$.
(Note that after this step the parents of every node of $G$ form a clique.)
Finally, make all edges of $G$ undirected.
\end{definition}

We shall require the following simple observation.

\begin{lemma}
\label{lem:moralization-complexity}
Given a Bayes net over $n$ nodes, its moralization can be computed in time $O\cprn{\cpoly{n}}$.
\end{lemma}

\begin{proof}
Let $B$ be a Bayes net over a DAG $G$ that has $n$ nodes.
Let $v$ be a node of $G$ and let $\Pi\cprn{v}$ be the set of the parents of $v$.
We can construct a clique among the nodes of $\Pi\cprn{v}$ in time $O\cprn{n^2}$, since $\abs{\Pi\cprn{v}}\leq n$.
Therefore we can construct all of the required cliques in time $n\cdot O\cprn{n^2}=O\cprn{n^3}$.
Finally, we can make all directed edges of $G$ undirected in time $O\cprn{n^2}$.
This yields a total running time of $O\cprn{n^3}$.
\end{proof}

\subsection{Total Variation Distance}

The following notion of distance is central in this work.

\begin{definition}[Total variation distance]
For probability distributions $P,Q$ over a finite sample space $\mathcal{D}$, the \emph{total variation distance} of $P$ and $Q$ is
\[
\dtv(P,Q)
=\max_{S \subseteq \mathcal{D}}\cprn{P(S)-Q(S)}.
\]
\end{definition}

Note that $\dtv(P,Q) = \frac{1}{2} \sum_{w\in\mathcal{D}}\abs{P\cprn{w}-Q\cprn{w}}$.
Equivalently,
\[
\dtv(P,Q)
= \sum_{w \in \mathcal{D}}\max(0, P(w) - Q(w))
= \sum_{w \in \mathcal{D}}(P(w) -\min(P(w), Q(w))).
\]

\subsection{Probabilistic Inference}

In this work, {\em probabilistic inference} is the following computational task:
Given (a representation of) random variables $X_1,\dots,X_n$ and (a representation of) sets $S_1,\dots,S_n$ such that for all $i$ the set $S_i$ is a subset of the range of $X_i$, compute $\cpr{X_1\in S_1, \dots, X_n\in S_n}$.%
\footnote{Note that a notion of probabilistic inference that has previously been considered~\cite{kwisthout2010necessity} is the following:
Given random variables $X_1,\dots,X_n$, a set $I = \{i_1, \cdots, i_k\}\subseteq\sqbra{n}$, values $x_{i_1},\dots,x_{i_k}$ that belong to the ranges of $X_{i_1},\dots,X_{i_k}$, respectively, and an event $E$, compute the probability $\cpr{\prn{X_{i_1},\dots,X_{i_k}}=\prn{x_{i_1},\dots,x_{i_k}}|E}$.
However, algorithms such as variable elimination for inference in this sense also work for the notion we consider. }

Let us now define probabilistic inference (oracle) queries.

\begin{definition}[Probabilistic inference query over Bayes nets]
A {\em probabilistic inference query} takes a description of a Bayes net distribution $P$ over $n$ nodes and alphabet size $\ell$ and descriptions of sets $S_1,\dots,S_n$, where for all $1\leq i\leq n$, $S_i \subseteq [\ell]$, and returns  the value of $\Prob_{P}[{X_1\in S_1, \dots, X_n\in S_n}]$ in one time step.
\end{definition}

\subsection{Treewidth and Tree Decompositions}

We require the definition of treewidth.

\begin{definition}
\label{def:tw}
A \emph{tree decomposition of an undirected graph $G = (V, E)$} is a tree $T$ with nodes $X_1, \dots, X_n$, where each $X_i$ is a subset of $V$, satisfying the following properties (the term node is used to refer to a vertex of $T$ to avoid confusion with vertices of $G$):
$(a)$
The union of all sets $X_i$ equals $V$.
That is, each graph vertex is contained in at least one tree node.
$(b)$
If $X_i$ and $X_j$ both contain a vertex $v$, then all nodes $X_k$ of $T$ in the (unique) path between $X_i$ and $X_j$ contain $v$ as well.
Equivalently, the tree nodes containing vertex $v$ form a connected subtree of $T$.
$(c)$
For every edge $(v_1, v_2)$ in the graph, there is a subset $X_i$ that contains both $v_1$ and $v_2$.
That is, vertices are adjacent in the graph only when the corresponding subtrees have a node in common.
The \emph{width of a tree decomposition} is the size of its largest set $X_i$ minus one.
The \emph{treewidth ${\rm tw}(G)$} is the minimum width among all possible tree decompositions of $G$.
\end{definition}

We shall also extend the notion of treewidth to Bayes nets, as follows.

\begin{definition}
The \emph{treewidth of a Bayes net} is defined to be equal to the treewidth of its moralization.
\end{definition}

We require the following two theorems, \Cref{thm:tree-decomposition} and \Cref{thm:variable-elimination}, respectively; \Cref{thm:tree-decomposition} is about a tree decomposition algorithm and \Cref{thm:variable-elimination} is about the variable elimination algorithm.

\begin{theorem}[Tree decomposition~\cite{RS84}]
\label{thm:tree-decomposition}
There is a $O\cprn{w3^{3w}n^2}$-time algorithm that finds a tree decomposition of width $4w + 1$, if the treewidth of the input graph is at most $w$.
\end{theorem}

We will make use of the variable elimination algorithm to efficiently implement probabilistic inference queries for bounded treewidth Bayes nets.

\begin{theorem}[Variable elimination; following Zhang and Poole~\cite{ZP94}]
\label{thm:variable-elimination}
There is an algorithm, called the \emph{variable elimination algorithm}, for the following task:
Given a Bayes net $B$ over variables $X_1,\dots,X_n\in\sqbra{\ell}$, sets $S_1,\dots,S_n\subseteq\sqbra{\ell}$, the moralization $M_B$ of $B$, and a tree decomposition $\mathcal{T}$ of width $w$ of $M_B$, compute the probability $\cprq{X_1\in S_1,\dots,X_n\in S_n}{B}$.
The running time of this algorithm is $O\cprn{n\ell^w}$.
\end{theorem}

\section{Technical Overview}

\label{sec:technical}

We present in this section some intuition regarding the technical aspects of our results. 

\subsection{Proof of \texorpdfstring{\Cref{thm:dtv-to-inference}}{}}

\label{sec:thm1over}

For the sake of simplicity of exposition, in this overview we assume that both the Bayes net distributions $P$ and $Q$ are defined over a directed path of length $n$ (over a finite alphabet of size $\ell$).
However, the ideas can be generalized to arbitrary Bayes nets.
Refer to \Cref{sec:preliminaries} for definitions and notation that we use here. 

Our approach relies on the well-known importance sampling technique.
The high-level approach is to define an estimator function $f$ and a distribution $\pi$ so that $\cexpectq{f}{\pi} = \dtv(P,Q)/Z$ where $Z$ is a normalization constant.

We start with the following characterization of $\dtv$:
$\dtv(P, Q) = \sum_{w} g^*(w)$ where $g^*(w)$ is defined to be $P(w) - \min(P(w), Q(w))$.
We define an auxiliary function $g(w)$ which is an overestimate of $g^*(w)$.
Define $h(w, i)$ as
\[
h(w, i) = \min(P_{i|i-1}(w_i|w_{i-1}), Q_{i|i-1}(w_i|w_{i-1}))
\]
and let $h(w) = \prod_{i=1}^n h(w, i)$.
Finally, set $g(w) = P(w) - h(w)$.
Comparing $g(w)$ with $g^*(w)$, note that $\min(P(w), Q(w)$ is the {\em minimum of the products} where product is of the form $\prod_i P_{i|i-1}(w_i|w_{i-1})$ (similarly for $Q$).
Whereas $h(w)$ is the product of the minimums.
Thus, it can be seen that $g(w) \geq g^*(w)$.

Recall that our goal is to estimate $\dtv(P, Q)=\sum_w g^*(w)$.
To estimate this we appeal to the classic importance sampling technique (for example, see~\cite{ChenS97}).
Define function $f(w) = \frac{g^*(w)}{g(w)}$ and distribution $\pi(w) = \frac{g(w)}{Z}$, where $Z = \sum_w g(w)$ is the normalizing constant.
Observe that
\[
\Exp_{w \sim \pi}\!\left[f(w)\right]
= \frac{\sum_w g^*(w)}{Z}
= \frac{\dtv(P, Q)}{Z}.
\]
Now the algorithm to estimate $\dtv$ works as follows:
Empirically estimate $\Exp_{\pi}[f(w)]$ by drawing samples from the distribution $\pi$ and multiply the empirical estimate by $Z$.
We appeal to the standard Chernoff bounds to obtain the guarantee on the quality of the approximation as well as the run time.

This approach will work if the following conditions are satisfied:
\begin{enumerate}
\item[(i)]
It is the case that $f(w)=\frac{g^*(w)}{g(w)}$ lies between $0$ and $1$.
This follows from the definition of $g(w)$ and the earlier discussion.
\item[(ii)]
The expectation $\Exp_{w \sim \pi}\!\left[f(w)\right ]$ is large enough (inverse polynomial) so that an additive approximation will lead to a multiplicative approximation.
\item[(iii)]
Sampling from the distribution $\pi$ and computing $Z = \sum_w g(w)$ can be done efficiently.
\end{enumerate}
The {\em key insight} that we bring in is that sampling from the distribution $\pi$ and computing $Z$ reduces to inference queries to a Bayes net distribution over the same directed graph as that of $P$ and $Q$.
Thus (iii) becomes efficient for all Bayes net distributions for which inference queries are feasible.
We also show that $\frac{1}{2n} \leq \Exp_{w \sim \pi}\!\left[f(w)\right ] $ thus satisfying (ii).
The rest of this subsection is devoted to explaining how inference queries can be used to compute $Z$.
We start with the known connection between $\dtv$ and {\em coupling} between distributions where the quantity $\min(P(w), Q(w))$ naturally arises.

\begin{definition}
\label{def:coupling}
Let $P$ and $Q$ are two arbitrary distributions on a common symbol set $\sqbra{\ell}$ (where $\ell>0$).
A \emph{coupling} of $P$ and $Q$ is a distribution on pairs $\prn{X,Y}$ such that $X\sim P$ and $Y\sim Q$.
An \emph{optimal coupling} of $P$ and $Q$ is a distribution on pairs $\prn{X,Y}$ such that (1) $X\sim P$, $Y\sim Q$, and (2) for any $w\in\sqbra{\ell}$, $\cpr{X=Y=w}=\min\cprn{P\cprn{w},Q\cprn{w}}$.
\end{definition}

It is well known that for any coupling $(X,Y)$ between $P$ and $Q$, $\dtv(P,Q) \leq \cpr{X \neq Y}$.
Additionally, for an optimal coupling as defined above, $\cpr{X\neq Y}$ exactly equals $\dtv(P,Q)$.
Note that,
\[
\cpr{X \neq Y}
= \sum_w \Prob[X = w, X \neq Y]
= \sum_w \prn{\Prob[X = w] - \Prob[X = Y = w]}.
\]
Thus the term $\Prob[X=w] - \Prob[X=Y=w]$ equals $P(w) - \Prob[X=Y=w]$.
By the definition of optimal coupling, $\Prob[X=Y=w]$ is precisely $\min(P(w), Q(w))$.
Thus $g^*(w) := P(w) - \min(P(w), Q(w))$ equals $\Pr[X = w \wedge X \neq Y]$.
Thus $g^*(w)$ has an interpretation using optimal coupling.

Since $g^*(w)$ can be expressed as the probability of an event over the optimal coupling, it is natural to ask whether there is a coupling $\mathcal{L}$ such that $\Prob_{\mathcal L}[X = w \wedge X \neq Y] = g(w)$. In a very recent work,~\cite{fgjw22} showed that when $P$ and $Q$ are product distributions, $g(w)$ admits such a characterization using couplings.
They define $\mathcal{L}$ as {\em local coupling} between $X$ and $Y$:
A joint distribution $\mathcal{L}$ on $(X,Y)=(X_1, \dots, X_n, Y_1, \dots, Y_n)$ where each $(X_i, Y_i)$ is independently sampled from an optimal coupling of the $i$-th marginals of $P$ and $Q$.

Generalizing this approach to Bayes nets poses several obstacles.
As in the case of product distributions, suppose we seek a coupling $\mathcal{L}$ of $P$ and $Q$ that also forms a Bayes net over the directed path.
In other words, we would like a coupling $\mathcal{L}$ generating the tuple $(X_1, \dots, X_n, Y_1, \dots, Y_n)$ such that each $(X_i, Y_i)$ is independent of $(X_{i-2}, Y_{i-2})$ conditioned on $(X_{i-1}, Y_{i-1})$.
However, there is an immediate problem:
Namely, $X_i$ and $X_{i-2}$ may be dependent given $X_{i-1}$ through the path $X_{i-2} \rightarrow Y_{i-1} \rightarrow X_i$, and similarly $Y_i$ and $Y_{i-2}$ may be dependent given $Y_{i-1}$ through the path $Y_{i-2} \rightarrow X_{i-1} \rightarrow Y_i$.
Hence, it may not be possible%
\footnote{Note that this issue does not arise for product distributions as there are no paths to speak of.}
to ensure that $(X_1, \dots, X_n)$ form a copy of $P$ and $(Y_1, \dots, Y_n)$ form a copy of $Q$, as is required for a coupling.

In light of obstacles faced by the coupling-based approach, we introduce a new notion of {\em partial couplings}.
The introduction of this notion is a primary conceptual contribution of our work.

\begin{definition}
\label{def:partial-coupling}
A \emph{ partial coupling} of distributions $P$ and $Q$ is a distribution on pairs $\prn{X,Y}$ such that (i) $X\sim P$ and (ii) it is the case that  $\cpr{X=Y=w} \leq \min\cprn{P\cprn{w},Q\cprn{w}}$.
\end{definition}

With the above definition in hand, we will show that it is possible to construct a partial coupling $\mathcal{L}$ of distributions $P$ and $Q$ such that $\mathcal{L}$ can be expressed as a Bayes net distribution over a graph that has the same {\em structure} as $P$ and $Q$.
We illustrate this for the case when $P$ and $Q$ are Bayes net distributions over a directed path.
We define a partial coupling $\mathcal{L}$ that is local.
The CPTs are defined as follows:
For any $b, c_1, c_2 \in [\ell]$,
\[
\Prob[X_i = Y_i = b | X_{i-1}=c_1, Y_{i-1}=c_2]
= \min(P_{i|i-1}(b|c_1), Q_{i|i-1}(b|c_2)).
\]
We will adjust the rest of the entries of the CPT and ensure that for all $b, c_1, c_2$, it is the case that $\Prob[X_i = b|X_{i-1}=c_1, Y_{i-1}=c_2] = P_{i|i-1}(b|c_1)$.
This ensures that $X\sim P$.
It can be shown that $\mathcal{L}$ is a {\em partial coupling}.

With this, we can indeed connect the function $g$ to the local partial coupling distribution $\mathcal{L}$.
We will show that $g(w) = \Prob_{\mathcal L}[X = w \wedge X \neq Y]$ and $Z = \Prob_{\mathcal L}[X \neq Y]$.
Note that $\Prob_{\mathcal L}[X \neq Y] = 1- \Prob_{\mathcal L}[X = Y]$.
Let $E_i$ denote the event that $(X_i, Y_i) \in \{(1, 1), (2, 2), \ldots (\ell, \ell)\}$.
Note that $\Prob_{\mathcal L}[X = Y]$ is $\Prob_{\mathcal L}[E_1 \cap E_2 \cap \ldots \cap E_n]$ which is an inference query to the Bayes net distribution $\mathcal{L}$.

To summarize:
We have shown that the quantity $Z$ can be computed by making an inference query to the distribution $\mathcal{L}$ which is expressible as a Bayes net over a straight line graph.
We can build on this idea to show that sampling from the distribution $\pi$ can also be done by making inference calls to the distribution $\mathcal{L}$.
What remains is to show that $\Exp_{\pi}[f(w)] = {\dtv(P, Q)}/{Z}$ is large enough.
We will establish that $Z \leq 2n\cdot\dtv(P, Q)$.
The proof of this inequality is somewhat technical and crucially uses properties in the definition of $\mathcal{L}$ mentioned above. 

\subsection{Proofs of the Rest of the Results}

We outline here the main proof ideas of the rest of our results.

\paragraph{Proof of \texorpdfstring{\Cref{thm:FPRAS-bounded-tw}}{}.}

The proof of \Cref{thm:FPRAS-bounded-tw} is an application of \Cref{thm:dtv-to-inference}.
To make use of \Cref{thm:dtv-to-inference}, we establish that probabilistic inference (i.e., computing the probability $\cprq{X_1\in S_1,\dots,X_n\in S_n}{}$) can be efficiently implemented for Bayes nets of constant alphabet size and logarithmic treewidth (\Cref{lem:inference-for-bounded-tw}).
It is known that a tree decomposition of graphs that have logarithmic treewidth can be computed in polynomial time~\cite{RS84}.
The variable elimination algorithm of~\cite{ZP94} shows that inference can be done in polynomial time given a tree decomposition, provided that the treewidth of the Bayes net is logarithmic in the dimension of the distribution.

\paragraph{Proof of \texorpdfstring{\Cref{thm:sharp-P-complete}}{}.}

\Cref{thm:sharp-P-complete} is proved by showing a reduction from $\#\mathrm{SAT}$ to computing the TV distance between an appropriately defined Bayes net and the uniform distribution.
This is achieved by creating a Bayes net that captures the circuit structure of a Boolean formula $F$ of which we want to compute its number of satisfying assignments.
The CPTs of this Bayes net mimic the function of the logical gates (AND, OR, NOT) of $F$.

\paragraph{Proof of \texorpdfstring{\Cref{thm:FPRAS-bayes-net-vs-uniform}}{}.}

\Cref{thm:FPRAS-bayes-net-vs-uniform} is proved by giving an algorithm that exploits the following property of TV distance.
Let $P$ be a Bayes net over $n$ variables that has maximum in-degree $d$ and alphabet size $\ell$.
In this case $\dtv\cprn{P,\U}$ is equal to
\begin{align*}
\frac{1}{2}\sum_x\abs{P\cprn{x}-\U\cprn{x}}
&=\sum_x\max\cprn{0,P\cprn{x}-\U\cprn{x}}\\
&=\sum_x\U\cprn{x}\max\cprn{0,\frac{P\cprn{x}}{\U\cprn{x}}-1}\\
&=\cexpectq{\max\cprn{0,\frac{P\cprn{x}}{\U\cprn{x}}-1}}{x\sim\U}\\
&=\cexpectq{\max\cprn{0,P\cprn{x}\ell^n-1}}{x\sim\U}.
\end{align*}
This yields a natural estimator for $\dtv\cprn{P,\U}$, whereby we draw samples $x_1,\dots,x_m\sim\U$ and then compute and output
\[
\frac{1}{m}\sum_{i=1}^m\max\cprn{0,P\cprn{x_i}\ell^n-1}.
\]
The crux of our analysis is the result that $\max\cprn{0,P\cprn{x}\ell^n-1}$ is between the values $0$ and $1+O\cprn{\dtv\cprn{P,\U}\ell^{d+1}n}$.
This enables us to use a value of $m$ that is in $O\cprn{\cpoly{n\ell^d,1/\varepsilon,\log\cprn{1/\delta}}}$, whereby $\varepsilon$ is the accuracy error and $\delta$ is the confidence error of the FPRAS.
Note that the running time is polynomial in the input length, as any description of the Bayes net $P$ has size at least $n+\ell^{d+1}$.

\section{From TV Distance to Probabilistic Inference}

\label{sec:dtv-to-inference}

In this section, we prove \Cref{thm:dtv-to-inference} and \Cref{thm:FPRAS-bounded-tw}.
In the following, let $T\cprn{G,\ell}$ be the running time of some implementation of a probabilistic inference oracle for a Bayes net over a DAG $G$ that has alphabet size $\ell$.
We will first state the formal version of \Cref{thm:dtv-to-inference}.

\begin{customthm}{\ref{thm:dtv-to-inference}}[Formal]
\label{thm:dtv-to-inference-formal}
There is a polynomial-time randomized algorithm that takes a DAG $G$, two Bayes nets $P$ and $Q$ over $G$ (as CPTs) with alphabet size $\ell$, and parameters $\varepsilon,\delta$ as inputs and behaves as follows.
The algorithm constructs a Bayes net distribution $\mathcal{L}$ over the same DAG $G$ with alphabet size $\ell^2$, makes probabilistic inference queries to $\mathcal{L}$, and outputs an $(1+\varepsilon)$-relative approximation of $\dtv(P,Q)$ with probability at least $1-\delta$.
The running time of this algorithm is $T\cprn{G,\ell^2}\cdot O\cprn{n^3\varepsilon^{-2}\ell\log\delta^{-1}}$ and the number of its probabilistic inference queries is $O\cprn{n^3\varepsilon^{-2}\ell\log\delta^{-1}}$.
\end{customthm}

The rest of the section is devoted to proving \Cref{thm:dtv-to-inference} and is organized as follows.
We first introduce the ingredients that are necessary for describing the algorithm (many of these are defined in \Cref{sec:technical} for path Bayes nets).
In \Cref{sec:implement}, we show how the algorithm can be implemented using probabilistic inference queries.
Finally, in \Cref{sec:analysis} we establish its correctness.

Let $P$ and $Q$ be two Bayes net distributions defined over a DAG $G$ with $n$ nodes and alphabet $[\ell]$.
Without loss of generality, assume that the nodes are topologically ordered as in the sequence $1,2,\ldots,n$.

Let $w$ be an element of the sample space, i.e., a $n$-symbol string over $[\ell]$.
Given $1 \leq i \leq n$, $\Pi(i)$ denotes the set of parents of $i$ in $G$ and let $w_{\Pi(i)}$ denote the projection $w$ at the parents of node $i$ in $G$.
We first define a function $h$ over $[\ell]^n \times [n]$ as follows:
\[
h(w, i)
:=\min\cprn{P_{i|\Pi\prn{i}}\cprn{w_i|w_{\Pi\prn{i}}},Q_{i|\Pi\prn{i}}\cprn{w_i|w_{\Pi\prn{i}}}}.
\]

\paragraph{Descriptions of $f$, $Z$, and $\pi$.}

The {\em estimator function} $f$ is defined as follows:
$f(w):= g^{*}(w)/g(w)$
where
\[
g^{*}(w) = P(w)-\min(P(w),Q(w))
\]
and
\[
g(w):=P(w)-\prod_{i=1}^n h\cprn{w, i}
\]
for all $w$.
It is straightforward to show that $f$ is computable in time $O\cprn{n}$.
We define $Z:=\sum_{w \in [\ell]^n} g(w)$ to be a normalization constant.
Finally, the distribution $\pi$ is specified by  the probability function $\pi(w):=g(w)/Z$ for all $w$.

\paragraph{Description of $\mathcal{L}$.}

We now define a Bayes net distribution $\mathcal{L}$ over the graph $G$ which is used to make inference queries by the algorithm.
The distribution $\mathcal{L}$ is over the alphabet $[\ell]^2$ and is a joint distribution $(X,Y)$ where $X$ and $Y$ take value over $[\ell]^n$.
We specify a CPT for $(X,Y)$.
For this, we need to specify for every $i$ and $b,z \in [\ell]$ the probability $\cpr{(X_i, Y_i) = (b, z)}$ conditioned on the values $\Pi(i)$ take.
We will first describe the probability where both $X_i$ and $Y_i$ take the same value $b$.
For every $c_1, c_2 \in \sqbra\ell^{\abs{\Pi\prn{i}}}$,
\[
\cpr{\prn{X_i,Y_i}
= \prn{b,b}|\prn{X_{\Pi\prn{i}},Y_{\Pi\prn{i}}} = \prn{c_1,c_2}}
= \min\cprn{P_{i|\Pi\prn{i}}\cprn{b|c_1},Q_{i|\Pi\prn{i}}\cprn{b|c_2}}.
\]
Define the remaining probabilities to ensure that the marginal $X$ is distributed according to $P$.
That is, for every $z\neq b$ set $\cpr{\prn{X_i,Y_i}=\prn{b,z}|\prn{X_{\Pi\prn{i}},Y_{\Pi\prn{i}}}=\prn{c_1,c_2}}$ so that the following holds:
\begin{align*}
& \sum_{z:z \neq b} \cpr{\prn{X_i,Y_i}=\prn{b,z}|\prn{X_{\Pi\prn{i}},Y_{\Pi\prn{i}}}=\prn{c_1,c_2}}\\
& \qquad= P_{i|\Pi\prn{i}}\cprn{b|c_1}- \min\cprn{P_{i|\Pi\prn{i}}\cprn{b|c_1},Q_{i|\Pi\prn{i}}\cprn{b|c_2}}
\end{align*}

\paragraph{Description of the Algorithm.}

Now we are ready to describe the algorithm (see \Cref{alg:FPRAS}).

\begin{algorithm}[htbt]
\caption{FPRAS for $\dtv$ estimation using a probabilistic inference oracle.}
\label{alg:FPRAS}
\begin{algorithmic}[1]
\REQUIRE Bayes nets $P,Q$ over DAG $G$ with $n$ nodes, parameters $\varepsilon,\delta$.
\ENSURE The output $\mathsf{Est}$ is an $(1+\varepsilon)$-approximation of $\dtv\cprn{P,Q}$, with probability at least $1-\delta$.
\STATE Construct the Bayes net distribution $\mathcal{L}$ over $G$
\STATE Compute $Z$ by making one {\em  inference query} to $\mathcal{L}$
\IF{$Z = 0$}
\RETURN $0$
\ENDIF
\STATE $m \gets Cn^2\varepsilon^{-2}\log\delta^{-1}$\qquad (for some large $C>0$)
\STATE $F \gets 0$
\FOR{$i\gets 1\ \mathbf{to}\ m$}
\STATE Sample $w^i\sim\pi$ by making {\em inference queries} to  $\mathcal{L}$
\STATE $F\gets F + f\cprn{w^i}$
\ENDFOR
\STATE $\mathsf{Est}\gets Z\cdot F/m$
\RETURN $\mathsf{Est}$
\end{algorithmic}
\end{algorithm}

\subsection{The Power of Probabilistic Inference}

\label{sec:implement}

This subsection is devoted to showing that the sampling from the distribution $\pi$ and the computation of the normalization constant $Z$ can be done by making probabilistic inference queries.
Recall that $\mathcal{L}$ is the joint distribution $(X, Y)$.
We start with the following crucial observation which states that the marginal $X$ (in $\mathcal{L}$) is distributed according to $P$.

\begin{observation}
\label{obs:XisP}
For every $b \in [\ell]$ and $c_1, c_2 \in \sqbra\ell^{\abs{\Pi\prn{i}}}$,
\[
\cpr{X_i = b|\prn{X_{\Pi\prn{i}},Y_{\Pi\prn{i}}} = \prn{c_1,c_2}}
=P_{i|\Pi\prn{i}}\cprn{b|c_1}.
\]   
\end{observation}

\begin{proof}
We have
\begin{align*}
\cpr{X_i = b|\prn{X_{\Pi\prn{i}},Y_{\Pi\prn{i}}} = \prn{c_1,c_2}}
& = \sum_{z \in[\ell]}
\cpr{\prn{X_i,Y_i}=\prn{b,z}|\prn{X_{\Pi\prn{i}},Y_{\Pi\prn{i}}}=\prn{c_1,c_2}}\\
& = \cpr{\prn{X_i,Y_i} = \prn{b,b}|\prn{X_{\Pi\prn{i}},Y_{\Pi\prn{i}}} = \prn{c_1,c_2}}\\
& \qquad +\sum_{z:z \neq b} \cpr{\prn{X_i,Y_i}=\prn{b,z}|\prn{X_{\Pi\prn{i}},Y_{\Pi\prn{i}}}=\prn{c_1,c_2}}\\
& = \min\cprn{P_{i|\Pi\prn{i}}\cprn{b|c_1},Q_{i|\Pi\prn{i}}\cprn{b|c_2}}\\
& \qquad + P_{i|\Pi\prn{i}}\cprn{b|c_1}-\min\cprn{P_{i|\Pi\prn{i}}\cprn{b|c_1},Q_{i|\Pi\prn{i}}\cprn{b|c_2}}\\
& = P_{i|\Pi\prn{i}}\cprn{b|c_1}
\end{align*}
as desired.
\end{proof}

Therefore, $X$ factorizes like $P$ with its conditional probabilities matching that of $P$ and hence $X\sim P$.
This realizes the notion of a local partial coupling as was earlier discussed in \Cref{sec:thm1over} and satisfies all three properties:
(i) $\mathcal{L}$ is a Bayes net distribution over the same DAG $G$ (that is used to describe distributions $P$ and $Q$),
(ii) $X \sim P$, and
(iii) in the joint distribution $\prn{X,Y}$, the conditional probabilities are equal to the minimum of the two conditional probabilities associated to $P$ and $Q$ as it is the case in standard couplings.

In \Cref{clm:Z=prob} we relate the normalization constant $Z$ of the distribution $\pi$ to the marginals $X$ and $Y$ of the distribution $\mathcal{L}$.
Moreover, we also relate the generalized normalization constant
\[
Z_{b_1,\dots,b_k}
:=\sum_{w:\prn{w_1,\dots,w_k}=\prn{b_1,\dots,b_k}}g\cprn{w},
\]
for $b_1,\dots,b_k\in\sqbra{\ell}$, to the marginals $X$ and $Y$ of the distribution $\mathcal{L}$.
We need this generalized normalization constant to show that sampling from the distribution $\pi$ (\Cref{clm:sample-pi}) can be efficiently done via probabilistic inference queries.

\begin{claim}
\label{clm:Z=prob}
It is the case that $Z=\cpr{X\neq Y}$ and
\[
Z_{b_1,\dots,b_k}
=\cpr{X\neq Y,X_1=b_1,\dots,X_k=b_k}
\]
for any $b_1,\dots,b_k\in\sqbra{\ell}$.
\end{claim}

\begin{proof}
Since $X\sim P$ and for that matter $P\cprn{w}=\cpr{X=w}$, we have
\begin{align*}
g(w)
& = P(w) - \prod_{i=1}^n \min\cprn{P_{i|\Pi\prn{i}}\cprn{w_i|w_{\Pi\prn{i}}}, Q_{i|\Pi\prn{i}}\cprn{w_i|w_{\Pi\prn{i}}}}\\
& = P(w) - \cpr{X=Y=w}\\
& = \cpr{X=w} - \cpr{X=Y=w}\\
& = \cpr{X=w} - \cpr{Y=w|X=w} \cdot \cpr{X=w}\\
& = \cpr{X=w} \prn{1 - \cpr{Y=w|X=w}} \\
& = \cpr{X=w} \cpr{Y \neq w | X = w} \\
& = \cpr{X = w, Y \neq w}.
\end{align*}
Therefore, we have that
\[
Z
=\sum_w g\cprn{w}
=\cpr{X\neq Y}
\]
and
\begin{align*}
Z_{b_1,\dots,b_k}
&=\sum_{w:\prn{w_1,\dots,w_k}=\prn{b_1,\dots,b_k}}g\cprn{w}\\
&=\sum_{w:\prn{w_1,\dots,w_k}=\prn{b_1,\dots,b_k}}\cpr{X=w,Y\neq w}\\
&=\cpr{X\neq Y,X_1=b_1,\dots,X_k=b_k}.
\qedhere
\end{align*}
\end{proof}

The following claim says that $Z$ and $Z_{b_1,\dots,b_k}$ can be easily computed given access to a probabilistic inference oracle for $\mathcal{L}$.

\begin{claim}
\label{clm:computation-of-Z}
It is the case that $Z$ and $Z_{b_1,\dots,b_k}$ (for any $b_1,\dots,b_k\in\sqbra{\ell}$) can be computed efficiently  by making $O\cprn{1}$ probabilistic inference queries to the Bayes net distribution $\mathcal{L}$.
\end{claim}

\begin{proof}
We will use \Cref{clm:Z=prob}.
Note that $Z=\cpr{X\neq Y}$ is equal to $1-\cpr{X=Y}$.
Therefore it suffices to compute $\cpr{X=Y}$ by using a probabilistic inference oracle.
This can done by observing that $\cpr{X=Y}$ is equal to $\cpr{\prn{X_1,Y_1}\in S_1,\dots,\prn{X_n,Y_n}\in S_n}$ for $S_1=\cdots=S_n=\brkts{\prn{1,1},\dots,\prn{\ell,\ell}}$.

Now note that the quantity $Z_{b_1,\dots,b_k}$ is equal to $\cpr{X\neq Y,X_1=b_1,\dots,X_k=b_k}$ which, in turn, is equal to $\cpr{X_1=b_1,\dots,X_k=b_k} -\cpr{X=Y,X_1=b_1,\dots,X_k=b_k}$.
What is left is to show how to compute these two probabilities by using a probabilistic inference oracle.
We have that $\cpr{X_1=b_1,\dots,X_k=b_k} = \cpr{\prn{X_1,Y_1}\in S_1,\dots,\prn{X_n,Y_n}\in S_n}$ for $S_i=\brkts{\prn{b_i,1},\dots,\prn{b_i,\ell}}$ for all $1\leq i\leq k$ and $S_{k+1}=\cdots=S_n=\sqbra{\ell}^2$.

Similarly, we have that
\[
\cpr{X=Y,X_1=b_1,\dots,X_k=b_k}
=\cpr{\prn{X_1,Y_1}\in S_1,\dots,\prn{X_n,Y_n}\in S_n}
\]
for $S_i=\brkts{\prn{b_i,b_i}}$ for all $1\leq i\leq k$, and $S_{k+1}=\cdots=S_n=\brkts{\prn{1,1},\dots,\prn{\ell,\ell}}$.
\end{proof}

We can now show that probabilistic inference queries allow for efficient sampling from $\pi$.

\begin{claim}
\label{clm:sample-pi}
Sampling from the distribution $\pi$ can be implemented in time $O(n\ell)$ by making $O(n\ell)$ probabilistic inference queries.
\end{claim}

\begin{proof}
We describe how to sample from $\pi$ iteratively, symbol by symbol.
Assume that we have sampled the first $k-1$ symbols, that is, assume that we have already sampled $w_1,\dots,w_{k-1}$ to be equal to $b_1,\dots,b_{k-1}\in\sqbra{\ell}$.
We describe now how to sample $w_k$.
For every possible value $b\in\sqbra{\ell}$ of $w_k$, we compute the marginal
\begin{align*}
\mu_b
:=\pi\cprn{b_1,\dots,b_{k-1},b}
=\frac{\sum_{w:\prn{w_1\cdots w_k}=\prn{b_1\cdots b_{k-1}b}}g\cprn{w}}{Z}
=\frac{Z_{b_1\cdots b_{k-1}b}}{Z}
\end{align*}
by two invocations of \Cref{clm:computation-of-Z}.
Then, we sample $w_k$ based on the values $\brkts{\mu_b}_{b=1}^\ell$.

Let $S\cprn{n}$ be the number of steps to sample $n$ symbols from $\pi$.
The above procedure gives the recurrence relation $S(n) = O(\ell) + S(n-1)$ which yields $S\cprn{n}=O\cprn{n\ell}$.
Since we perform at most two probabilistic inference queries for every coordinate and every symbol, the total number of probabilistic inference queries is equal to $S\cprn{n}=O\cprn{n\ell}$.
\end{proof}

\subsection{Analysis of the Algorithm}

\label{sec:analysis}

Next, we establish some useful properties of the function $f$ and the distribution $\pi$.

\begin{claim}
\label{clm:f-bounded}
For any $w$, it is the case that $0\leq f(w)\leq 1$.
\end{claim}

\begin{proof}
We separately show $0\leq f(w)$ and $f(w)\leq 1$.
To establish $0\leq f(w)$, since the numerator is non-negative, it suffices to show that $g(w)=P(w)-\prod_{i=1}^n h\cprn{w, i}\geq0$ or equivalently $P(w)\geq\prod_{i=1}^n h\cprn{w, i}$.

We have
\[
P(w)
= \prod_{i=1}^n P_{i|\Pi\prn{i}}\cprn{w_i|w_{\Pi\prn{i}}}
\geq \prod_{i=1}^n \min\cprn{P_{i|\Pi\prn{i}}\cprn{w_i|w_{\Pi\prn{i}}},Q_{i|\Pi\prn{i}}\cprn{w_i|w_{\Pi\prn{i}}}},
\]
which is equal to $\prod_{i=1}^n h(w, i)$, by the definition of $h$.

For showing $f(w)\leq 1$, it suffices to show that $P(w)-Q(w)\leq g(w)$ (since $0/g\cprn{w}=0\leq1$).
Since $g(w) = P(w) - \prod_{i-1}^n h(w, i)$, it suffices to show that $Q(w) \geq \prod_{i=1}^n h(w, i)$.
An argument identical to the above, where we showed that $P(w) \geq \prod_{i=1}^n h(w, i)$, will show this.
\end{proof}

We will also relate the expected value of the function $f$ with respect to the distribution $\pi$ to $\dtv(P, Q)$.

\begin{claim}
\label{clm:f-exp}
It is the case that $\cexpectq{f(w)}{\pi}=\dtv(P,Q)/Z$.
\end{claim}

\begin{proof}
We have that $\cexpectq{f(w)}{\pi}$ is equal to
\begin{align*}
\cexpectq{\frac{\max\cprn{0,P(w)-Q(w)}}{g(w)}}{\pi}
&=\sum_w\pi\cprn{w}\frac{\max\cprn{0,P(w)-Q(w)}}{g(w)}\\
&=\sum_w\frac{g(w)}{Z}\frac{\max\cprn{0,P(w)-Q(w)}}{g(w)}\\
&=\frac{1}{Z}\sum_w\max\cprn{0,P(w)-Q(w)}\\
&=\frac{\dtv\cprn{P,Q}}{Z}.
\qedhere
\end{align*}
\end{proof}

We need the following lemma that ensures the estimand is large enough to facilitate Monte Carlo sampling.

\begin{lemma}
\label{lem:Z-upper-bound}
It is the case that $Z\leq2n\cdot\dtv(P, Q)$.
\end{lemma}

\begin{proof}
By \Cref{clm:Z=prob}, it suffices to show that $\cpr{X \neq Y} \leq 2n\cdot \dtv(P,Q)$.
We split the event $(X\neq Y)$ into $n$ disjoint events $\brkts{E_i}_{i=1}^n$.
Without loss of generality, assume that $1,2,\ldots,n$ is the topological ordering of the vertices of $G$.
Event $E_i$ is defined as $(\bigwedge_{1\leq j\leq i-1} X_j= Y_j) \wedge (X_i \neq Y_i)$.
Note that the $E_i$'s are disjoint.
Thus $\pr{X \neq Y}= \sum_{i} \pr{E_i}$.
We have that
\[
\pr{E_i}
\leq \pr{(X_i \neq Y_i) \wedge (X_{\Pi(i)} = Y_{\Pi(i)})}\\
= \sum_{\sigma} \pr{(X_i \neq Y_i) \wedge (X_{\Pi(i)}, Y_{\Pi(i)}) = (\sigma, \sigma)}
\]
where $\sigma$ is an assignment for $\Pi(i)$ (note that the length of $\sigma$ is equal to the in-degree of $i$).
Henceforth, for notational brevity, we shall omit the dependence on $i$.
Thus,
\begin{align*}
\cprq{X\neq Y}{}
& = \sum_{i}\pr{E_i} \\
& \leq \sum_i\sum_\sigma\cprq{X_i\neq Y_i\land\prn{X_{\Pi\prn{i}},Y_{\Pi\prn{i}}}=\prn{\sigma,\sigma}}{} \\
& = \sum_i\sum_\sigma\cprq{X_i\neq Y_i|\prn{X_{\Pi\prn{i}},Y_{\Pi\prn{i}}}=\prn{\sigma,\sigma}}{}
\cprq{\prn{X_{\Pi\prn{i}},Y_{\Pi\prn{i}}}=\prn{\sigma,\sigma}}{}.
\end{align*}
We require the following claim.

\begin{claim}
\label{clm:difference-P-Q}
For any $\sigma$, $\cprq{X_i\neq Y_i|\prn{X_{\Pi\prn{i}},Y_{\Pi\prn{i}}}=\prn{\sigma,\sigma}}{} = \dtv\cprn{P_{i|\Pi\prn{i}}\cprn{\cdot|\sigma},Q_{i|\Pi\prn{i}}\cprn{\cdot|\sigma}}$.
\end{claim}

\begin{proof}
We have that $\cprq{X_i\neq Y_i|\prn{X_{\Pi\prn{i}},Y_{\Pi\prn{i}}}=\prn{\sigma,\sigma}}{}$ is
\begin{align*}
&1-\cprq{X_i=Y_i|\prn{X_{\Pi\prn{i}},Y_{\Pi\prn{i}}}=\prn{\sigma,\sigma}}{}\\
&\qquad=1-\sum_{c\in[\ell]}\cprq{\prn{X_i,Y_i}=\prn{c,c}|\prn{X_{\Pi\prn{i}},Y_{\Pi\prn{i}}}=\prn{\sigma,\sigma}}{}\\
&\qquad=1-\sum_{c\in[\ell]}\min\cprn{P_{i|\Pi\prn{i}}\cprn{c|\sigma},Q_{i|\Pi\prn{i}}\cprn{c|\sigma}}\\
&\qquad=\sum_{c\in[\ell]}P_{i|\Pi\prn{i}}\cprn{c|\sigma}-\sum_{c\in[\ell]}\min\cprn{P_{i|\Pi\prn{i}}\cprn{c|\sigma},Q_{i|\Pi\prn{i}}\cprn{c|\sigma}}\\
&\qquad=\sum_{c\in[\ell]}\prn{P_{i|\Pi\prn{i}}\cprn{c|\sigma}-\min\cprn{P_{i|\Pi\prn{i}}\cprn{c|\sigma},Q_{i|\Pi\prn{i}}\cprn{c|\sigma}}}\\
&\qquad=\sum_{c\in[\ell]}\max\cprn{0,P_{i|\Pi\prn{i}}\cprn{c|\sigma}-Q_{i|\Pi\prn{i}}\cprn{c|\sigma}}\\
&\qquad=\dtv\cprn{P_{i|\Pi\prn{i}}\cprn{\cdot|\sigma},Q_{i|\Pi\prn{i}}\cprn{\cdot|\sigma}}.
\qedhere
\end{align*}
\end{proof}

By \Cref{clm:difference-P-Q} we have that $\cprq{X\neq Y}{}$ is at most
\begin{align*}
&\sum_i\sum_\sigma\cprq{\prn{X_{\Pi\prn{i}},Y_{\Pi\prn{i}}}=\prn{\sigma,\sigma}}{}
\dtv\cprn{P_{i|\Pi\prn{i}}\cprn{\cdot|\sigma},Q_{i|\Pi\prn{i}}\cprn{\cdot|\sigma}}\\
&\qquad\leq\sum_i\sum_\sigma\cprq{X_{\Pi\prn{i}}=\sigma}{}\frac{1}{2}\sum_c\abs{P_{i|\Pi\prn{i}}\cprn{c|b}-Q_{i|\Pi\prn{i}}\cprn{c|\sigma}}\\
&\qquad\leq\sum_i\sum_\sigma P_{\Pi\prn{i}}(\sigma)\frac{1}{2}\sum_c\abs{P_{i|\Pi\prn{i}}\cprn{c|\sigma}-Q_{i|\Pi\prn{i}}\cprn{c|\sigma}}~~\mbox{(since $X \sim P$ by \Cref{obs:XisP})}\\
&\qquad=\sum_i\sum_\sigma\frac{1}{2}\sum_c\abs{P_{\Pi\prn{i}}\cprn{\sigma}P_{i|\Pi\prn{i}}\cprn{c|\sigma}-P_{\Pi\prn{i}}\cprn{\sigma}Q_{i|\Pi\prn{i}}\cprn{c|\sigma}}\\
&\qquad=\sum_i\sum_\sigma\frac{1}{2}\sum_c\left|P_{\Pi\prn{i}}\cprn{\sigma}P_{i|\Pi\prn{i}}\cprn{c|\sigma}
-Q_{\Pi\prn{i}}\cprn{\sigma}Q_{i|\Pi\prn{i}}\cprn{c|\sigma}\right.\\
&\qquad\qquad\left.+Q_{\Pi\prn{i}}\cprn{\sigma}Q_{i|\Pi\prn{i}}\cprn{c|\sigma}
-P_{\Pi\prn{i}}\cprn{\sigma}Q_{i|\Pi\prn{i}}\cprn{c|\sigma}\right|\\
&\qquad\leq\sum_i\sum_\sigma\frac{1}{2}\sum_c\abs{P_{\Pi\prn{i}}\cprn{\sigma}P_{i|\Pi\prn{i}}\cprn{c|\sigma}
-Q_{\Pi\prn{i}}\cprn{\sigma}Q_{i|\Pi\prn{i}}\cprn{c|\sigma}}\\
&\qquad\qquad+\sum_i\sum_\sigma\frac{1}{2}\sum_c
\left|Q_{\Pi\prn{i}}\cprn{\sigma}Q_{i|\Pi\prn{i}}\cprn{c|\sigma}
-P_{\Pi\prn{i}}\cprn{\sigma}Q_{i|\Pi\prn{i}}\cprn{c|\sigma}\right|\\
&\qquad=\sum_i\sum_\sigma\frac{1}{2}\sum_c\abs{P_{i|\Pi\prn{i}}\cprn{c|\sigma}
-Q_{i|\Pi\prn{i}}\cprn{c|\sigma}}\\
&\qquad\qquad+\sum_i\sum_\sigma\frac{1}{2}\abs{Q_{\Pi\prn{i}}\cprn{\sigma}
-P_{\Pi\prn{i}}\cprn{\sigma}}\sum_c Q_{i|\Pi\prn{i}}\cprn{c|\sigma}\\
&\qquad=\sum_i\frac{1}{2}\sum_\sigma\sum_c\abs{P_{i|\Pi\prn{i}}\cprn{c|\sigma}
-Q_{i|\Pi\prn{i}}\cprn{c,\sigma}}\\
&\qquad\qquad+\sum_i\sum_\sigma\frac{1}{2}\abs{Q_{\Pi\prn{i}}\cprn{\sigma}
-P_{\Pi\prn{i}}\cprn{\sigma}}\\
&\qquad=\sum_i\dtv\cprn{P_{i|\Pi\prn{i}},Q_{i|\Pi\prn{i}}}+\sum_i\dtv\cprn{P_{i},Q_{i}}\\
&\qquad\leq2n\cdot\dtv(P,Q).
\end{align*}
The last inequality follows because the inequalities
\[
\dtv\cprn{P_{i|\Pi\prn{i}},Q_{i|\Pi\prn{i}}}\leq\dtv\cprn{P,Q}
\qquad
\dtv\cprn{P_{i},Q_{i}}\leq\dtv\cprn{P,Q}
\]
hold.
\end{proof}

We are now ready to prove the correctness of and provide a running time bound for \Cref{alg:FPRAS}.
We have, from Hoeffding's inequality (\Cref{lem:Hoeffding}), that
\begin{align*}
\cpr{\abs{\mathsf{Est}-\dtv\cprn{P,Q}}>\varepsilon\dtv\cprn{P,Q}}
&=\cpr{\abs{\frac{Z}{m}\sum_{i=1}^mf\cprn{w^i}-Z\cexpectq{f\cprn{w}}{\pi}}>\varepsilon\dtv\cprn{P,Q}}\\
&\qquad=\cpr{\abs{\sum_{i=1}^mf\cprn{w^i}-m\cexpectq{f\cprn{w}}{\pi}}>\frac{m\varepsilon}{Z}\dtv\cprn{P,Q}}\\
&\qquad\leq2\exp\!\left(-\frac{2m^2\varepsilon^2\dtv^2\cprn{P,Q}}{Z^2\sum_{i=1}^m\prn{0-1}^2}\right)\\
&\qquad\leq2\exp\!\left(-\frac{2m^2\varepsilon^2\dtv^2\cprn{P,Q}}{4n^2\dtv^2\cprn{P,Q}m}\right)\\
&\qquad=2\exp\!\left(-\frac{m\varepsilon^2}{2n^2}\right),
\end{align*}
which is at most $\delta$ whenever $m=\Omega\cprn{n^2\varepsilon^{-2}\log \delta^{-1}}$.
The second inequality follows from \Cref{lem:Z-upper-bound}.

Thus  the running time of \Cref{alg:FPRAS} is $O\cprn{mn\ell}$, which equals $O\cprn{n^3\varepsilon^{-2}\ell\log \delta^{-1}}$, since we draw $m$ samples from $\pi$, we can sample from $\pi$ in time $O\cprn{n\ell}$, and evaluate $f$ in time $O\cprn{n}$.
Finally, the number of probabilistic inference queries is at most $O\cprn{n^3\varepsilon^{-2}\ell\log \delta^{-1}}$.

\subsection{Application: Bayes Nets of Small Treewidth}

\label{sec:FPRAS-bounded-tw}

We are now ready to prove \Cref{thm:FPRAS-bounded-tw}.

\begin{customthm}{\ref{thm:FPRAS-bounded-tw}}[Formal]
\label{thm:FPRAS-bounded-tw-formal}
There is an FPRAS for estimating the TV distance between two Bayes nets of treewidth $w=O\cprn{\log n}$ and alphabet size $\ell=O\cprn{1}$, which are defined over the same DAG of $n$ nodes.
In particular, if $\varepsilon$ and $\delta$ are the accuracy and confidence errors of the FPRAS, respectively, the FPRAS runs in time $\cpoly{n}\cdot O\cprn{\varepsilon^{-2}\log\delta^{-1}}$.
\end{customthm}

The proof of \Cref{thm:FPRAS-bounded-tw-formal} will follow from the lemma below, \Cref{lem:inference-for-bounded-tw}, and \Cref{thm:dtv-to-inference} for $\ell=O\cprn{1}$ and $T\cprn{G,\ell^2}=O\cprn{\cpoly{n}}$.

\begin{lemma}
\label{lem:inference-for-bounded-tw}
Probabilistic inference is efficient for all Bayes nets over $n$ variables which have alphabet size $\ell=O\cprn{1}$ and treewidth $O\cprn{\log n}$.
\end{lemma}

\begin{proof}
Let $B$ be a Bayes net over variables $X_1,\dots,X_n$ that has alphabet size $\ell=O\cprn{1}$ and treewidth $w=O\cprn{\log n}$.
Let $S_1,\dots,S_n\subseteq\sqbra{\ell}$ be sets.
The probabilistic inference task that we want to perform is to compute the probability $\cprq{X_1\in S_1,\dots,X_n\in S_n}{B}$.

First, we construct the moralization of $B$ (see \Cref{def:moralization}), namely $M_B$, in time $O\cprn{\cpoly{n}}$ by invoking \Cref{lem:moralization-complexity}.
Then, we use \Cref{thm:tree-decomposition} to compute a tree decomposition $\mathcal{T}$ of $M_B$ of width at most $4w+1\leq5w$ in time $O\cprn{w3^{3w}n^2}$.
Finally, we use the variable elimination algorithm of \Cref{thm:variable-elimination} on $B$, $S_1,\dots,S_n$, $M_B$, and $\mathcal{T}$ to compute $\cprq{X_1\in S_1,\dots,X_n\in S_n}{B}$ in time $O\cprn{n\ell^{5w}}$.

The running time of this procedure is $O\cprn{\cpoly{n}}+O\cprn{w3^{3w}n^2}+O\cprn{n\ell^{5w}}=O\cprn{\cpoly{n}}$, whereby we have used the facts that $\ell=O\cprn{1}$ and $w=O\cprn{\log n}$.
This concludes the proof.
\end{proof}

The proof of \Cref{thm:FPRAS-bounded-tw-formal} now follows by invoking \Cref{thm:dtv-to-inference-formal} for $\ell=O\cprn{1}$ and $T\cprn{G,\ell^2}=O\cprn{\cpoly{n}}$.

\section{TV Distance Between a Bayes Net and the Uniform Distribution}

Here, we prove \Cref{thm:sharp-P-complete} and \Cref{thm:FPRAS-bayes-net-vs-uniform}.

\subsection{\texorpdfstring{$\#\P$}{\#P}-Completeness}

\label{sec:hardness}

The main result of this subsection is \Cref{thm:sharp-P-complete}.
Recall that a function $f$ from $\{0,1\}^*$ to non-negative integers is in the class $\#\P$ if there is a polynomial time non-deterministic Turing machine $M$ so that for any $x$, it is the case that $f(x)$ is equal to the number of accepting paths of $M(x)$.

We now prove \Cref{thm:sharp-P-complete}.

\begin{proof}[Proof of \texorpdfstring{\Cref{thm:sharp-P-complete}}{}]
In what follows, we separately show membership in $\#\P$ and $\#\P$-hardness.

\subparagraph{Membership in \texorpdfstring{$\#\P$}{}.}

Let $P$ be a Bayes net distribution over the Boolean domain $\bool^n$.
The goal is to design a nondeterministic machine ${\cal N}$ so that the number of accepting paths of ${\cal N}$ (normalized by an appropriate quantity) equals $\dtv(P,\U)$.
We will assume that the probabilities specified in the CPTs of the Bayes net for $P$ are fractions.
Let $M$ be equal to $2^n$ times the product of the denominators of all the probabilities in the CPTs.
The non-deterministic machine ${\cal N}$ first guesses an element $i\in\bool^n$ in the sample space of $P$, computes $|P(i)-1/2|$ by using the CPTs, then guesses an integer $0\leq z\leq M$, and finally accepts if and only if $1\leq z\leq M|P(i)-1/2|$.
(Note that $M|P(i)-1/2|=\abs{M\cdot P(i)-M/2}$ is an integer.)
It follows that
\[
\dtv(P,\U)
=\frac{1}{2}\sum_{i\in\bool^n}\abs{P\cprn{i}-\frac{1}{2}}
=\frac{\text{number of accepting paths of }{\cal N}}{2M}
\]
since the number of accepting paths of ${\cal N}$ is equal to $\sum_{i\in\bool^n}\prn{M\abs{P\cprn{i}-1/2}}$ which is equal to $M\sum_{i\in\bool^n}\abs{P\cprn{i}-1/2}$, or $2M\dtv(P,Q)$.

\subparagraph{\texorpdfstring{$\#\P$}{}-Hardness.}

For the $\#\P$-hardness part, the proof gives a Turing reduction from the problem of counting the satisfying assignments of a CNF formula (which is $\#\P$-hard to compute) to computing the total variation distance between a Bayes net distribution and the uniform distribution.
In what follows, by a graph of a formula we mean the DAG that captures the circuit structure of $F$, whereby the nodes are either AND, OR, NOT, or variable gates, and the edges correspond to wires connecting the gates.

Let $F$ be a CNF formula viewed as a Boolean circuit.
Assume $F$ has $n$ input variables $x_1,\ldots, x_n$ and $m$ gates $\Gamma=\brkts{y_1,\ldots,y_m}$, where $\Gamma$ is topologically sorted with $y_m$ being the output gate.
We will define a Bayes net distribution on some DAG $G$ which, intuitively, is the graph of $F$.

The vertex set of $G$ is split into two sets $\mathcal{X}$ and $\mathcal{Y}$, and a node $Z$.
The set ${\mathcal X}=\brkts{X_i}_{i=1}^n$ contains $n$ nodes with node $X_i$ corresponding to variable $x_i$ and the set ${\mathcal Y}=\brkts{Y_i}_{i=1}^m$ contains $m$ nodes with each node $Y_i$ corresponding to gate $y_i$.
So totally there are $n+m+1$ nodes.
There is a directed edge from node $V_i$ to node $V_j$ if the gate/variable corresponding to $V_i$ is an input to $V_j$.

The distribution $P$ on $G$ is given by a CPT defined as follows.
Each $X_i$ is a uniformly random bit.
For each $Y_i$, its CPT is deterministic:
For each of the setting of the parents $Y_j,Y_k$, namely $y_j,y_k$, the variable $Y_i$ takes the value of the gate $y_i$ for that setting of its inputs $y_j,y_k$.
Finally, let $Z$ be the value of $Y_m$ OR-ed with a random bit.

Note that the formula $F$ computes a Boolean function on the input variables.
Let $f:\bool^n\to\bool$ be this function.
We extend $f$ to $\{0,1\}^{m}$ (i.e., $f:\bool^n\to\bool^m$) to also include the values of the intermediate gates.

With this notation in mind, for any binary string $XYZ$ of length $n+m+1$ it is the case that $P$ has a probability $0$ if $Y \neq f(X)$.
Let $A:=\brkts{x\mid F\cprn{x}=1}$ and $R:=\brkts{x\mid F\cprn{x}=0}$.

To finish the proof, we will write the number of satisfying assignments of $F$, namely $\abs{A}$, as a polynomial-time computable function of $\dtv(P,\U)$:
We have
\begin{align*}
2\cdot\dtv(P,\U)
=\sum_{X,Y,Z}\abs{P-\U}
=\underbrace{\sum_{\substack{X,Y,Z\\ Y\neq f(X)}}|P-\U|}_{(1)}
+\underbrace{\sum_{\substack{X,Y,Z\\ Y = f(X)}} |P-\U|}_{(2)}
\end{align*}
where we have abused the notation $P,\U$ to denote the probabilities $P\cprn{X,Y,Z},\U\cprn{X,Y,Z}$.
We will calculate $(1)$ and $(2)$ separately.
For $(1)$ we have:
\begin{align*}
{\sum_{\substack{X,Y,Z\\ Y\neq f(X)}}|P-\U|}
={\sum_{\substack{X,Y,Z\\ Y\neq f(X)}}\abs{0-\frac{1}{2^{n+m+1}}}}
=\frac{2^{n+1}(2^m-1)}{2^{n+m+1}}
=1-\frac{1}{2^m}.
\end{align*}
For $(2)$, we have
\begin{align*}
{\sum_{\substack{X,Y,Z \\ Y = f(X)}} |P-\U|}
&=\underbrace{\sum_{\substack{X,f(X),Z \\ X \in A}} |P-\U|}_{(3)}
+\underbrace{\sum_{\substack{X,f(X),Z \\ X \in R}} |P-\U|}_{(4)}
\end{align*}
and now we calculate the terms $(3)$ and $(4)$ separately.
For $(3)$, we have:
\begin{align*}
{\sum_{\substack{X,f(X),Z \\ X \in A}} |P-\U|}
&={\sum_{\substack{X,f(X),0 \\ X \in A}} |P-\U|}
+{\sum_{\substack{X,f(X),1 \\ X \in A}} |P-\U|}\\
&={\sum_{\substack{X,f(X),0 \\ X \in A}} \abs{0-\frac{1}{2^{n+m+1}}}}
+{\sum_{\substack{X,f(X),1 \\ X \in A}} \abs{\frac{1}{2^n}-\frac{1}{2^{n+m+1}}}}\\
&=\frac{|A|}{2^{n+m+1}} + \frac{|A| \cdot (2^{m+1} -1)}{2^{n+m+1}}\\
&=\frac{|A|}{2^n}
\end{align*}
and for $(4)$ we have
\begin{align*}
\sum_{\substack{X,f(X),Z \\ X \in R}} |P-\U|
&=\sum_{\substack{X,f(X),0 \\ X \in R}} |P-\U|
+\sum_{\substack{X,f(X),1 \\ X \in R}} |P-\U|\\
&=\sum_{\substack{X,f(X),0 \\ X \in R}} \abs{\frac{1}{2^{n+1}}-\frac{1}{2^{n+m+1}}}
+\sum_{\substack{X,f(X),1 \\ X \in R}} \abs{\frac{1}{2^{n+1}}-\frac{1}{2^{n+m+1}}}\\
&=\frac{|R|\cdot(2^m-1)\cdot 2}{2^{n+m+1}}.
\end{align*}
Thus 
\begin{align*}
2\cdot\dtv(P,\U)
&=(1) + (2) \\
&=(1) + (3) + (4) \\
&=1-\frac{1}{2^m} + \frac{|A|}{2^n} + \frac{|R|\cdot (2^m-1) \cdot 2}{2^{n+m+1}} \\
&=2\prn{1-\frac{1}{2^m} + \frac{|A|}{2^{m+n+1}}}
\end{align*}
since $\abs{A}+\abs{R}=2^n$.
For that matter, $\dtv(P,\U)=\frac{|A|}{2^{n+m+1}}+\prn{1-\frac{1}{2^m}}$ or
\[
\abs{A}=2^{n+m+1}\prn{\dtv(P,\U)-\prn{1-\frac{1}{2^m}}}.
\]
That concludes the proof.
\end{proof}

\subsection{Estimation in Randomized Polynomial Time}

\label{sec:Bayes-nets-vs-Uniform}

We prove \Cref{thm:FPRAS-bayes-net-vs-uniform}.

\begin{customthm}{\ref{thm:FPRAS-bayes-net-vs-uniform}}[Formal]
\label{thm:FPRAS-bayes-net-vs-uniform-formal}
There is an FPRAS for estimating the TV distance between a Bayes net $P$ and the uniform distribution.
Let $n$ be the number of nodes of $P$, let $\ell$ be the size of its alphabet, and let $d$ be its maximum in-degree.
Then the running time of this FPRAS is $O\cprn{n^3\ell^{2d+2}\varepsilon^{-2}\log\delta^{-1}}$ whereby $\varepsilon$ is the accuracy error and $\delta$ is the confidence error of the FPRAS.
\end{customthm}

\begin{remark}
Note that the running time of the FPRAS of \Cref{thm:FPRAS-bayes-net-vs-uniform-formal} is polynomial in the input length, as the description of the Bayes net $P$ in terms of the CPTs has size at least $n+\ell^{d+1}$.
\end{remark}

We shall now prove \Cref{thm:FPRAS-bayes-net-vs-uniform-formal}.
We require the following lemma (which we will prove below).

\begin{lemma}
\label{lem:ratio}
For all $x$, it is the case that
\[
1-O\cprn{\dtv\cprn{P,\U}\ell^{d+1}n}
\leq P\cprn{x}\ell^n
\leq1+O\cprn{\dtv\cprn{P,\U}\ell^{d+1}n}
\]
whenever $\dtv\cprn{P,\U}\leq\frac{1}{16\ell^{d+1}}$.
\end{lemma}

The proof of \Cref{thm:FPRAS-bayes-net-vs-uniform} now resumes as follows.
First, let us assume that $\dtv\cprn{P,\U}\le \frac{1}{16\ell^{d+1}}$ so that \Cref{lem:ratio} holds.
We have that $\dtv\cprn{P,\U}$ is equal to
\begin{align*}
\frac{1}{2}\sum_x\abs{P\cprn{x}-\U\cprn{x}}
&=\sum_x\max\cprn{0,P\cprn{x}-\U\cprn{x}}\\
&=\sum_x\U\cprn{x}\max\cprn{0,\frac{P\cprn{x}}{\U\cprn{x}}-1}\\
&=\cexpectq{\max\cprn{0,\frac{P\cprn{x}}{\U\cprn{x}}-1}}{x\sim\U}\\
&=\cexpectq{\max\cprn{0,P\cprn{x}\ell^n-1}}{x\sim\U}.
\end{align*}
This yields a natural estimator for $\dtv\cprn{P,\U}$, namely $\mathsf{Est}$, as follows:
\begin{enumerate}
\item
Sample $x_1,\dots,x_m\sim\U$ for some value of $m$ that we will fix later;
\item
compute $\max\cprn{0,P\cprn{x_i}\ell^n-1}$ for all $1\leq i\leq m$;
\item
output $\prn{1/m}\sum_{i=1}^m\max\cprn{0,P\cprn{x_i}\ell^n-1}$.
\end{enumerate}
We will now prove the correctness and upper bound the running time of this procedure.
We have from Hoeffding's inequality (\Cref{lem:Hoeffding}) and \Cref{lem:ratio} that
\begin{align*}
&\cpr{\abs{\mathsf{Est}-\dtv\cprn{P,\U}}>\varepsilon\dtv\cprn{P,\U}}\\
&\qquad=\cpr{\abs{\frac{1}{m}\sum_{i=1}^m\max\cprn{0,P\cprn{x_i}\ell^n-1}-\cexpectq{\max\cprn{0,P\cprn{x}\ell^n-1}}{x\sim\U}}>\varepsilon\dtv\cprn{P,\U}}\\
&\qquad=\cpr{\abs{\sum_{i=1}^m\max\cprn{0,P\cprn{x_i}\ell^n-1}-m\cexpectq{\max\cprn{0,P\cprn{x}\ell^n-1}}{x\sim\U}}>m\varepsilon\dtv\cprn{P,\U}}\\
&\qquad\leq 2\exp\!\left(-\frac{2m^2\varepsilon^2\dtv^2\cprn{P,\U}}{\sum_{i=1}^m\prn{0-O(\dtv\cprn{P,\U}\ell^{d+1}n)}^2}\right)\\
&\qquad=2\exp\!\left(-\frac{2m^2\varepsilon^2\dtv^2\cprn{P,\U}}{m\cdot O(\dtv^2\cprn{P,\U}\ell^{2d+2}n^2)}\right)\\
&\qquad=2\exp\!\left(-\frac{m\varepsilon^2}{O\cprn{\ell^{2d+2}n^2}}\right),
\end{align*}
which is at most $\delta$ whenever $m=\Omega\cprn{n^2\ell^{2d+2}\varepsilon^{-2}\log\delta^{-1}}$.

The running time of this procedure is $O\cprn{mn}=O\cprn{n^3\ell^{2d+2}\varepsilon^{-2}\log\delta^{-1}}$, since we draw $m$ samples and $P$ can be evaluated on any sample in time $O\cprn{n}$.

If $\dtv\prn{P,\U}>\frac{1}{16\ell^{d+1}}$, then it suffices to additively approximate $\dtv\prn{P,\U}$ up to error $\varepsilon/\prn{16\ell^{d+1}}$.
This can be done by Monte Carlo sampling using $m=\Omega\cprn{\ell^{2d+2}\varepsilon^{-2}\log\delta^{-1}}$ samples and $O(mn)=O\cprn{n\ell^{2d+2}\varepsilon^{-2}\log\delta^{-1}}$ time.

We now prove \Cref{lem:ratio}.

\begin{proof}[Proof of \Cref{lem:ratio}]
Let us denote the maximum in-degree of $P$ by $d$.
Let $X_0$ be an arbitrary node with its parents as $X_1,\dots,X_d$.

We have that $\gamma:=\dtv\cprn{P,\U}$ is at least
\begin{align*}
&\dtv\cprn{\prn{X_0,\dots,X_d},\prn{Y_0,\dots,Y_d}}\\
&\qquad=\frac{1}{2}\sum_{v_0}\cdots\sum_{v_d}\abs{\cpr{\prn{X_0,\dots,X_d}=\prn{v_0,\dots,v_d}}
-\cpr{\prn{Y_0,\dots,Y_d}=\prn{v_0,\dots,v_d}}}\\
&\qquad=\frac{1}{2}\sum_{v_0}\cdots\sum_{v_d}\abs{\cpr{X_0=v_0|X_1=v_1,\dots,X_d=v_d}\cpr{X_1=v_1,\dots,X_d=v_d}-\frac{1}{\ell^{d+1}}}
\end{align*}
or
\[
\frac{1}{2}\sum_{v_0}\cdots\sum_{v_d}
\abs{\cpr{X_0=v_0|X_1=v_1,\dots,X_d=v_d}\cpr{X_1=v_1,\dots,X_d=v_d}-\frac{1}{\ell^{d+1}}}
=\gamma
\]
or
\begin{align}
\abs{\cpr{X_0=v_0|X_1=v_1,\dots,X_d=v_d}\cpr{X_1=v_1,\dots,X_d=v_d}-\frac{1}{\ell^{d+1}}}
\leq2\gamma,
\label{eq:TV-2}
\end{align}
for any $v_0,\dots,v_d$.
We observe the following.

\begin{claim}
\label{clm:prob-Xi-2}
We have that $1/\ell^d-\gamma\leq\cpr{X_1=v_1,\dots,X_d=v_d}\leq1/\ell^d+\gamma$.
\end{claim}

\begin{proof}
Since $\dtv\cprn{P,\U}=\gamma$ and $\cpr{Y_1=v_1,\dots,Y_d=v_d}=1/\ell^d$, the claim is immediate.
\end{proof}

By \Cref{eq:TV-2} and \Cref{clm:prob-Xi-2} we have the following.

\begin{corollary}
For $\gamma<1/\prn{2\ell^d}$ we have that
\[
\abs{\cpr{X_0=v_0|X_1=v_1,\dots,X_d=v_d}-1/\ell}
\leq8\gamma \ell^d.
\]
\end{corollary}

\begin{proof}
By \Cref{eq:TV-2} we have
\[
\frac{1}{\ell^{d+1}}-2\gamma
\leq\cpr{X_0=v_0|X_1=v_1,\dots,X_d=v_d}\cpr{X_1=v_1,\dots,X_d=v_d}
\leq\frac{1}{\ell^{d+1}}+2\gamma
\]
or
\[
\frac{\frac{1}{\ell^{d+1}}-2\gamma}{\cpr{X_1=v_1,\dots,X_d=v_d}}
\leq\cpr{X_0=v_0|X_1=v_1,\dots,X_d=v_d}
\leq\frac{\frac{1}{\ell^{d+1}}+2\gamma}{\cpr{X_1=v_1,\dots,X_d=v_d}}
\]
or, by making use of \Cref{clm:prob-Xi-2},
\[
\frac{\frac{1}{\ell^{d+1}}-2\gamma}{\frac{1}{\ell^d}+\gamma}
\leq\cpr{X_0=v_0|X_1=v_1,\dots,X_d=v_d}
\leq\frac{\frac{1}{\ell^{d+1}}+2\gamma}{\frac{1}{\ell^d}-\gamma}
\]
or
\[
\frac{\frac{1}{\ell}-2\ell^d\gamma}{1+\ell^d\gamma}
\leq\cpr{X_0=v_0|X_1=v_1,\dots,X_d=v_d}
\leq\frac{\frac{1}{\ell}+2\ell^d\gamma}{1-\ell^d\gamma}.
\]
We now have
\begin{align*}
\cpr{X_0=v_0|X_1=v_1,\dots,X_d=v_d}
&\leq\frac{\frac{1}{\ell}+2\ell^d\gamma}{1-\ell^d\gamma}\\
&\leq\prn{\frac{1}{\ell}+2\ell^d\gamma}\prn{1+2\ell^d\gamma}\\
&=\frac{1}{\ell}+2\ell^{d-1}\gamma+2\ell^d\gamma+4\ell^{2d}\gamma^2\\
&\leq\frac{1}{\ell}+2\ell^{d}\gamma+2\ell^d\gamma+4\ell^{d}\gamma\\
&=\frac{1}{\ell}+8\ell^d\gamma,
\end{align*}
since $1/\prn{1-x}\leq1+2x$ for $x<1/2$ (here $x=\ell^d\gamma<1/2$), and
\begin{align*}
\cpr{X_0=v_0|X_1=v_1,\dots,X_d=v_d}
&\geq\frac{\frac{1}{\ell}-2\ell^d\gamma}{1+\ell^d\gamma}\\
&\geq\prn{\frac{1}{\ell}-2\ell^d\gamma}\prn{1-\ell^d\gamma}\\
&=\frac{1}{\ell}-\ell^{d-1}\gamma-2\ell^d\gamma+2\ell^{2d}\gamma^2\\
&\geq\frac{1}{\ell}-\ell^{d}\gamma-2\ell^d\gamma\\
&\geq\frac{1}{\ell}-8\gamma\ell^d,
\end{align*}
since $1/\prn{1+x}\geq1-x$ for $x<1/2$ (here $x=\ell^d\gamma<1/2$).
\end{proof}

The result now follows from the observation that
\[
\prn{1/\ell-8\gamma\ell^d}^n
\leq P\cprn{x}
=\prod_{i=1}^n\cpr{X_i=x_i|X_{\Pi\prn{X_i}}=x_{\Pi\prn{X_i}}}
\leq\prn{1/\ell+8\gamma\ell^d}^n
\]
or
\[
\prn{1-8\gamma\ell^{d+1}}^n
\leq P\cprn{x}\ell^n
\leq\prn{1+8\gamma\ell^{d+1}}^n
\]
or
\[
1-16\gamma\ell^{d+1}n
\leq P\cprn{x}\ell^n
\leq1+16\gamma\ell^{d+1}n,
\]
whereby we used the facts that $(1-\alpha)^k \ge (1-2\alpha k)$ and $(1+\alpha)^k \le (1+2\alpha k)$ whenever $\alpha<1/2$ and $k>0$, and the fact that $\gamma<1/\prn{16\ell^{d+1}}$ or $8\gamma\ell^{d+1}<1/2$.

Finally, we have
\[
1-16\dtv\cprn{P,\U}\ell^{d+1}n
\leq P\cprn{x}\ell^n
\leq1+16\dtv\cprn{P,\U}\ell^{d+1}n,
\]
as desired.
\end{proof}

\section{Conclusion}

\label{sec:conclusion}

We have established a general connection between probabilistic inference and TV distance computation.
In particular, we proved that TV distance estimation can be reduced to probabilistic inference in a structure preserving manner.
This enables us to prove the existence of a novel FPRAS for estimating the TV distance between Bayes nets of small treewidth.

The notion of {\em partial couplings} introduced in this work is of independent interest.
It would be fruitful to explore applications of this notion in other contexts.

We outline the following open problems:
Can we prove similar results for TV distance estimation between undirected graphical models?
Another problem of interest is to study other notions of distance, such as Wasserstein metrics.

\section*{Acknowledgements}

The work of AB was supported in part by National Research Foundation Singapore under its NRF Fellowship Programme (NRF-NRFFAI-2019-0002) and an Amazon Faculty Research Award.
The work of SG was supported by an initiation grant from IIT Kanpur and a SERB award CRG/2022/007985.
Pavan's work is partly supported by NSF award 2130536.
Vinodchandran's work is partly supported by NSF award 2130608.
This work was supported in part by National Research Foundation Singapore under its NRF Fellowship Programme [NRF-NRFFAI1-2019-0004] and an Amazon Research Award.
Part of the work was done during Meel, Pavan, and Vinodchandran's visit to the Simons Institute for the Theory of Computing.

\newcommand{\etalchar}[1]{$^{#1}$}


\begin{thebibliography}{KBvdG10}

\bibitem[BGM{\etalchar{+}}23]{BGMMPV23}
Arnab Bhattacharyya, Sutanu Gayen, Kuldeep~S. Meel, Dimitrios Myrisiotis, A.~Pavan, and N.~V. Vinodchandran.
\newblock On approximating total variation distance.
\newblock In {\em Proc. of IJCAI}, pages 3479--3487. ijcai.org, 2023.

\bibitem[BGMV20]{0001GMV20}
Arnab Bhattacharyya, Sutanu Gayen, Kuldeep~S. Meel, and N.~V. Vinodchandran.
\newblock Efficient distance approximation for structured high-dimensional distributions via learning.
\newblock In {\em Proc. of NeurIPS}, 2020.

\bibitem[BKM17]{blei2017variational}
David~M. Blei, Alp Kucukelbir, and Jon~D. McAuliffe.
\newblock Variational inference: A review for statisticians.
\newblock {\em Journal of the American Statistical Association}, 112(518):859--877, 2017.

\bibitem[Coo90]{cooper1990computational}
Gregory~F Cooper.
\newblock The computational complexity of probabilistic inference using {B}ayesian belief networks.
\newblock {\em Artificial intelligence}, 42(2-3):393--405, 1990.

\bibitem[CR14]{CR14}
Cl{\'{e}}ment~L. Canonne and Ronitt Rubinfeld.
\newblock Testing probability distributions underlying aggregated data.
\newblock In Javier Esparza, Pierre Fraigniaud, Thore Husfeldt, and Elias Koutsoupias, editors, {\em Automata, Languages, and Programming - 41st International Colloquium, {ICALP} 2014, Copenhagen, Denmark, July 8-11, 2014, Proceedings, Part {I}}, volume 8572 of {\em Lecture Notes in Computer Science}, pages 283--295. Springer, 2014.

\bibitem[CS97]{ChenS97}
Ming-Hui Chen and Qi-Man Shao.
\newblock On {M}onte {C}arlo methods for estimating ratios of normalizing constants.
\newblock {\em The Annals of Statistics}, 25(4):1563--1594, 1997.

\bibitem[CT06]{cover2006elements}
Thomas~M. Cover and Joy~A. Thomas.
\newblock {\em Elements of information theory {(2.} ed.)}.
\newblock Wiley, 2006.

\bibitem[Dec99]{dechter1999bucket}
Rina Dechter.
\newblock Bucket elimination: A unifying framework for reasoning.
\newblock {\em Artificial Intelligence}, 113(1-2):41--85, 1999.

\bibitem[Doe38]{Doe38}
Wolfang Doeblin.
\newblock Expos{\'e} de la th{\'e}orie des cha{\i}nes simples constantes de markova un nombre fini d’{\'e}tats.
\newblock {\em Math{\'e}matique de l’Union Interbalkanique}, 2(77-105):78--80, 1938.

\bibitem[Dwo06]{dwork2006differential}
Cynthia Dwork.
\newblock Differential privacy.
\newblock In Michele Bugliesi, Bart Preneel, Vladimiro Sassone, and Ingo Wegener, editors, {\em Automata, Languages and Programming, 33rd International Colloquium, {ICALP} 2006, Venice, Italy, July 10-14, 2006, Proceedings, Part {II}}, volume 4052 of {\em Lecture Notes in Computer Science}, pages 1--12. Springer, 2006.

\bibitem[FGJW23]{fgjw22}
Weiming Feng, Heng Guo, Mark Jerrum, and Jiaheng Wang.
\newblock A simple polynomial-time approximation algorithm for the total variation distance between two product distributions.
\newblock {\em TheoretiCS}, 2, 2023.

\bibitem[FLL24]{feng2023ondeterministically}
Weiming Feng, Liqiang Liu, and Tianren Liu.
\newblock On deterministically approximating total variation distance.
\newblock In David~P. Woodruff, editor, {\em Proceedings of the 2024 {ACM-SIAM} Symposium on Discrete Algorithms, {SODA} 2024, Alexandria, VA, USA, January 7-10, 2024}, pages 1766--1791. {SIAM}, 2024.

\bibitem[GSCM21]{golia2021designing}
Priyanka Golia, Mate Soos, Sourav Chakraborty, and Kuldeep~S. Meel.
\newblock Designing samplers is easy: The boon of testers.
\newblock In {\em Formal Methods in Computer Aided Design, {FMCAD} 2021, New Haven, CT, USA, October 19-22, 2021}, pages 222--230. {IEEE}, 2021.

\bibitem[GSV99]{GoldreichSV99}
Oded Goldreich, Amit Sahai, and Salil~P. Vadhan.
\newblock Can statistical zero knowledge be made non-interactive? or {O}n the relationship of {SZK} and {NISZK}.
\newblock In {\em Proc. of CRYPTO}, pages 467--484, 1999.

\bibitem[HBWP13]{hoffman2013stochastic}
Matthew~D. Hoffman, David~M. Blei, Chong Wang, and John Paisley.
\newblock Stochastic variational inference.
\newblock In {\em JMLR}, volume~14, pages 1303--1347, 2013.

\bibitem[HdBM20]{holtzen2020dice}
Steven Holtzen, Guy~Van den Broeck, and Todd~D. Millstein.
\newblock Dice: Compiling discrete probabilistic programs for scalable inference.
\newblock {\em CoRR}, abs/2005.09089, 2020.

\bibitem[KBCK23]{klinkenberg2023exactbayesian}
Lutz Klinkenberg, Christian Blumenthal, Mingshuai Chen, and Joost{-}Pieter Katoen.
\newblock Exact {B}ayesian inference for loopy probabilistic programs.
\newblock {\em CoRR}, abs/2307.07314, 2023.

\bibitem[KBvdG10]{kwisthout2010necessity}
Johan Kwisthout, Hans~L Bodlaender, and Linda~C van~der Gaag.
\newblock The necessity of bounded treewidth for efficient inference in {B}ayesian networks.
\newblock In {\em ECAI}, volume 215, pages 237--242, 2010.

\bibitem[KF09]{koller2009probabilistic}
Daphne Koller and Nir Friedman.
\newblock {\em Probabilistic graphical models: Principles and techniques}.
\newblock MIT press, 2009.

\bibitem[Lin02]{Lin02}
Torgny Lindvall.
\newblock {\em Lectures on the coupling method}.
\newblock Courier Corporation, 2002.

\bibitem[LMP01]{littman2001stochastic}
Michael~L Littman, Stephen~M Majercik, and Toniann Pitassi.
\newblock Stochastic {B}oolean satisfiability.
\newblock {\em Journal of Automated Reasoning}, 27:251--296, 2001.

\bibitem[LPW06]{LevinPeresWilmer2006}
David~A. Levin, Yuval Peres, and Elizabeth~L. Wilmer.
\newblock {\em {Markov chains and mixing times}}.
\newblock American Mathematical Society, 2006.

\bibitem[LS88]{lauritzen1988local}
Steffen~L Lauritzen and David~J Spiegelhalter.
\newblock Local computations with probabilities on graphical structures and their application to expert systems.
\newblock {\em Journal of the Royal Statistical Society: Series B (Methodological)}, 50(2):157--194, 1988.

\bibitem[Min01]{minka2001expectation}
Thomas Minka.
\newblock Expectation propagation for approximate {B}ayesian inference.
\newblock In {\em UAI}, 2001.

\bibitem[MT12]{MT12}
Sean~P Meyn and Richard~L Tweedie.
\newblock {\em Markov chains and stochastic stability}.
\newblock Springer Science \& Business Media, 2012.

\bibitem[MU05]{mitzenmacher2005probability}
Michael Mitzenmacher and Eli Upfal.
\newblock {\em Probability and Computing: Randomized Algorithms and Probabilistic Analysis}.
\newblock Cambridge University Press, 2005.

\bibitem[MWJ13]{murphy2013loopy}
Kevin~P. Murphy, Yair Weiss, and Michael~I. Jordan.
\newblock Loopy belief propagation for approximate inference: An empirical study.
\newblock {\em CoRR}, abs/1301.6725, 2013.

\bibitem[Pea88]{pearl1988probabilistic}
Judea Pearl.
\newblock {\em Probabilistic reasoning in intelligent systems: {N}etworks of plausible inference}.
\newblock Morgan Kaufmann, 1988.

\bibitem[RGB14]{ranganath2014black}
Rajesh Ranganath, Sean Gerrish, and David~M. Blei.
\newblock Black box variational inference.
\newblock {\em Arxiv preprint arXiv:1401.0118}, 2014.

\bibitem[RM15]{rezende2015variational}
Danilo~J. Rezende and Shakir Mohamed.
\newblock Variational inference with normalizing flows.
\newblock {\em Arxiv preprint arXiv:1505.05770}, 2015.

\bibitem[Rot96]{roth1996hardness}
Dan Roth.
\newblock On the hardness of approximate reasoning.
\newblock {\em Artificial Intelligence}, 82(1-2):273--302, 1996.

\bibitem[RS84]{RS84}
Neil Robertson and Paul~D. Seymour.
\newblock Graph minors {III.} {P}lanar tree-width.
\newblock {\em J. Comb. Theory, Ser. {B}}, 36(1):49--64, 1984.

\bibitem[SB14]{shwartz2014understanding}
Shai Shalev{-}Shwartz and Shai Ben{-}David.
\newblock {\em Understanding Machine Learning - From Theory to Algorithms}.
\newblock Cambridge University Press, 2014.

\bibitem[SRM21]{saad2021sppl}
Feras~A. Saad, Martin~C. Rinard, and Vikash~K. Mansinghka.
\newblock {SPPL:} probabilistic programming with fast exact symbolic inference.
\newblock In Stephen~N. Freund and Eran Yahav, editors, {\em {PLDI} '21: 42nd {ACM} {SIGPLAN} International Conference on Programming Language Design and Implementation, Virtual Event, Canada, June 20-25, 2021}, pages 804--819. {ACM}, 2021.

\bibitem[Sti95]{stinson1995cryptography}
Douglas~R. Stinson.
\newblock {\em Cryptography - theory and practice}.
\newblock Discrete mathematics and its applications series. {CRC} Press, 1995.

\bibitem[SV03]{SV03}
Amit Sahai and Salil~P. Vadhan.
\newblock A complete problem for statistical zero knowledge.
\newblock {\em J. {ACM}}, 50(2):196--249, 2003.

\bibitem[Vad12]{vadhan2012pseudorandomness}
Salil~P. Vadhan.
\newblock Pseudorandomness.
\newblock {\em Found. Trends Theor. Comput. Sci.}, 7(1-3):1--336, 2012.

\bibitem[WJ{\etalchar{+}}08]{wainwright2008graphical}
Martin~J Wainwright, Michael~I Jordan, et~al.
\newblock Graphical models, exponential families, and variational inference.
\newblock {\em Foundations and Trends{\textregistered} in Machine Learning}, 1(1--2):1--305, 2008.

\bibitem[ZMO23]{zaiser2023exact}
Fabian Zaiser, Andrzej~S. Murawski, and Chih{-}Hao~Luke Ong.
\newblock Exact bayesian inference on discrete models via probability generating functions: {A} probabilistic programming approach.
\newblock In Alice Oh, Tristan Naumann, Amir Globerson, Kate Saenko, Moritz Hardt, and Sergey Levine, editors, {\em Advances in Neural Information Processing Systems 36: Annual Conference on Neural Information Processing Systems 2023, NeurIPS 2023, New Orleans, LA, USA, December 10 - 16, 2023}, 2023.

\bibitem[ZP94]{ZP94}
Nevin~Lianwen Zhang and David~L. Poole.
\newblock A simple approach to {B}ayesian network computations, 1994.

\end{thebibliography}

\end{document}